\documentclass[twocolumn,twoside]{IEEEtran}

\ifCLASSINFOpdf

\else

\fi

\ifCLASSOPTIONcompsoc
    \usepackage[caption=false, font=normalsize, labelfont=sf, textfont=sf]{subfig}
\else
    \usepackage[caption=false, font=normalsize]{subfig}
\fi
\usepackage{lipsum}%
\usepackage[dvipsnames]{xcolor}
\usepackage{algorithm,algorithmic}
\usepackage{balance}
\usepackage{multicol}   
\usepackage{cite}
\usepackage{gensymb}
\usepackage{multirow}
\usepackage{graphics}
\usepackage{epsfig}
\usepackage{graphicx}
\usepackage{epstopdf}
\usepackage{textcomp}
\usepackage{amsmath}
\usepackage{mathtools}
\usepackage{filecontents}
\usepackage{lipsum,color}
\usepackage{amssymb}
\usepackage{float}
\usepackage{colortbl} 
\usepackage{tabularx}
\usepackage{times} 
\usepackage{amsthm}  

\usepackage{amsfonts}

\usepackage{amsmath, amssymb, amsthm}

\newtheorem{proposition}{Proposition}

\theoremstyle{break}

\begin{document}
\title{Compact Analytical Model for Real-Time Evaluation of OAM-Based Inter-Satellite Links}

\author{Mohammad~Taghi~Dabiri,~and~Mazen~Hasna,~{\it Senior Member,~IEEE}
	\thanks{Mohammad Taghi Dabiri and Mazen Hasna are with the Department of Electrical Engineering, Qatar University, Doha, Qatar (e-mail: m.dabiri@qu.edu.qa; hasna@qu.edu.qa).}

	\thanks{This publication was made possible by NPRP14C-0909-210008 from the Qatar Research, Development and Innovation (QRDI) Fund (a member of  Qatar Foundation). }	
}

\maketitle
\begin{abstract}
This paper presents an efficient analytical framework for evaluating the performance of inter-satellite communication systems utilizing orbital angular momentum (OAM) beams under pointing errors. An accurate analytical model is first developed to characterize intermodal crosstalk caused by beam misalignment in OAM-based inter-satellite links. Building upon this model, we derive efficient expressions to analyze and optimize system performance in terms of bit error rate (BER). Unlike traditional Monte Carlo-based methods that are computationally intensive, the proposed approach offers accurate performance predictions. 
This enables a substantial decrease in computation time while maintaining high accuracy, thanks to the use of analytical expressions for both crosstalk and BER.
This fast and accurate evaluation capability is particularly critical for dynamic low Earth orbit (LEO) satellite constellations, where network topology and channel conditions change rapidly, requiring real-time link adaptation. Furthermore, we systematically design and evaluate asymmetric OAM mode sets, which significantly outperform symmetric configurations in the presence of pointing errors. Our results also reveal key insights into the interaction between beam divergence, tracking accuracy, and link distance, demonstrating that the proposed framework enables real-time optimization of system parameters with high fidelity. The analytical findings are rigorously validated against extensive Monte Carlo simulations, confirming their practical applicability for high-mobility optical wireless systems such as LEO satellite networks.
\end{abstract}

\IEEEpeerreviewmaketitle

\section{Introduction}
\subsection{Background and Motivation}
The rapid advancement of global communication infrastructure is increasingly relying on Low Earth Orbit (LEO) satellite constellations, which are expected to play a central role in various strategic applications, including satellite-based internet (e.g., Starlink), Earth observation, global navigation systems (GNSS), and deep-space communication \cite{dabiri2024modulating,dabiri2024modulating2}. These large-scale constellations demand ultra-high-capacity, low-latency, and resilient inter-satellite communication links capable of supporting dynamic and frequently changing topologies \cite{abdulwahid2024inter}.

Orbital Angular Momentum (OAM) multiplexing has recently emerged as a promising enabler for such high-performance optical satellite links, owing to its unique helical phase structure that allows for orthogonal spatial mode multiplexing and significantly enhanced spectral efficiency \cite{noor2022review, chen2020orbital}.
In contrast to traditional radio-frequency, millimeter-wave, and terahertz communication systems, OAM-based free-space optics (FSO) offer a scalable and secure physical-layer solution with minimal spectrum congestion and inherent spatial diversity, which is particularly advantageous in the spaceborne domain \cite{elsayed2024performance}.
Despite its potential, the practical deployment of OAM systems in LEO constellations faces several key technical challenges. These include extreme sensitivity to pointing errors due to narrow beam profiles, significant mode-dependent loss and intermodal crosstalk under misalignment, and system complexity associated with beam generation, multiplexing, and alignment in high-mobility scenarios. Without addressing these issues, the benefits of OAM in realistic orbital environments remain largely theoretical.

This study is motivated by the need to bridge this gap between theory and practice. We develop an analytical framework that not only models the intermodal crosstalk induced by pointing errors but also enables accurate and efficient evaluation of system performance in terms of bit error rate (BER). By providing both analytical insights and practical design guidelines—such as optimal mode set selection and divergence control—this work contributes to the realization of robust, high-throughput OAM-based inter-satellite links tailored to the dynamic and challenging conditions of future LEO constellations.

\subsection{Literature Review}
Although extensive research has been conducted on analyzing the effects of atmospheric turbulence on OAM links, studies on the impact of pointing errors on mode crosstalk in OAM systems have been limited to works such as \cite{lavery2017free, ren2016experimental, pang2018qpsk, wang2021learning, gong2021recognition, ramakrishnan2024mitigation, gangwar2023mitigation, willner2021causes, xie2015performance, hu2023aiming, zhao2018spiral, alababneh2019crosstalk, elamassie2023characterization, elamassie2024modeling, djordjevic2023effect, li2024misalignment, elmeadawy2022hybrid, zhang2022information, raza2023data, mabena2023beam}.

The experimental studies in \cite{lavery2017free, ren2016experimental, pang2018qpsk} shed light on the impact of misalignment on OAM system performance under real-world conditions. In \cite{lavery2017free}, a 1.6~km free-space link showed that weak aberrations caused vortex splitting in higher-order OAM modes, degrading phase purity essential for spatial multiplexing and quantum key distribution (QKD). Similarly, \cite{ren2016experimental} revealed significant power fluctuations and crosstalk in a 400~Gbit/s OAM-multiplexed link over 120~m without laser tracking. In \cite{pang2018qpsk}, small apertures and lateral misalignments in a 400~Gbit/s link resulted in power loss and increased crosstalk, underscoring the importance of precise alignment for reliable OAM communications.

The studies in \cite{wang2021learning, gong2021recognition} emphasize the challenges posed by pointing errors in OAM-based FSO systems. In \cite{wang2021learning}, a machine learning-based alignment-free fractal interferometer was proposed for robust mode recognition under severe misalignment, while \cite{gong2021recognition} utilized convolutional neural networks to enhance OAM mode recognition by addressing the combined effects of pointing errors and limited apertures. Beamforming, as demonstrated in \cite{ramakrishnan2024mitigation}, improves signal quality, reduces error vector magnitudes, and enhances QPSK signal constellations. The radial shearing self-interferometric method in \cite{gangwar2023mitigation} mitigates lateral misalignment in the generation of perfect vortex beams. A comprehensive review in \cite{willner2021causes} discusses adaptive optics for real-time phase correction and spatial mode pre-compensation to enhance robustness. The review also highlights digital signal processing techniques such as multiple-input multiple-output (MIMO) equalization, which effectively suppresses crosstalk and maintains signal integrity in dynamic conditions like UAV-based links. Despite these advancements \cite{lavery2017free, ren2016experimental, pang2018qpsk, wang2021learning, gong2021recognition, ramakrishnan2024mitigation, gangwar2023mitigation, willner2021causes}, a unified mathematical framework to model and optimize the effects of pointing errors remains absent, limiting the development of scalable and robust OAM communication systems.

In \cite{xie2015performance}, the performance of OAM-based FSO links under pointing errors was analyzed, with a focus on power loss and crosstalk caused by fixed lateral displacements and angular misalignments, providing key metrics and design insights for such systems. Similarly, \cite{hu2023aiming} showed that angular tilt leads to spectral broadening and power leakage. The study in \cite{zhao2018spiral} examined the effect of pointing errors on the spiral spectrum of Laguerre-Gaussian (LG) beams, emphasizing lateral displacement and angular inclination, though it did not address advanced mitigation strategies or adaptive techniques for practical scenarios.
In \cite{alababneh2019crosstalk}, crosstalk in slightly misaligned FSO interconnects was analyzed using a novel diffraction model based on the cylindrical form of the Collins diffraction integral. Approximate closed-form expressions for the optical field of LG beams propagating through lens-based FSO interconnects with finite circular apertures were derived. More comprehensive statistical modeling was presented in \cite{elamassie2023characterization} and \cite{elamassie2024modeling}, where Beta and Generalized Gamma distributions were proposed to characterize crosstalk. According to these studies, fixed displacements beyond 6 mm lead to significant crosstalk.
While \cite{elamassie2023characterization} and \cite{elamassie2024modeling} provided valuable foundational analyses, their statistical models are limited. They overlooked critical parameters such as beam width and receiver aperture size, which this study identifies as highly influential. Even small variations in these parameters can substantially impact the crosstalk model, emphasizing the need for a more comprehensive approach to account for these factors in OAM system design and analysis.

Another category of research focused on statistical modeling and optimization of OAM-based systems under pointing errors. In \cite{djordjevic2023effect}, the impact of pointing errors on the BER of OAM-based FSO systems was analyzed, demonstrating how quasi-cyclic low-density parity-check (LDPC) codes combined with multidimensional constellations mitigated degradation under moderate turbulence. Similarly, \cite{li2024misalignment} investigated state-dependent misalignment and turbulence in OAM-based high-dimensional QKD, and proposed an atmospheric channel model to optimize secure key rate (SKR) and quantum bit error rate (QBER) under finite-key effects.
In \cite{elmeadawy2022hybrid}, a hybrid OAM-Multi Pulse-Position Modulation (M-PPM) scheme was introduced to counter pointing errors and turbulence, showing significant BER improvements over standalone OAM or MPPM techniques. Closed-form BER expressions revealed enhanced power and spectral efficiency even under moderate and strong turbulence. Similarly, \cite{zhang2022information} provided analytical evaluations of information capacity loss in LG beam-based optical links, accounting for pointing errors and turbulence, and identified optimal beam parameters to reduce performance degradation.
In practical demonstrations, \cite{raza2023data} achieved 400 Gbps per wavelength using LG modes and a 4-PAM modulation scheme with intensity modulation and direct detection (IMDD), reporting acceptable BERs at distances of 150 m, 500 m, and 700 m. Furthermore, \cite{mabena2023beam} investigated the beam quality factor for aberrated LG beams under astigmatism and spherical aberrations, revealing how beam width and LG mode order influenced quadratic, cubic, or quartic degradation patterns based on the aberration type.
These studies offered valuable insights into the design and optimization of OAM systems under practical challenges, though further work was required to unify these findings into a comprehensive framework for analyzing and mitigating pointing errors across diverse scenarios.

However, it should be noted that the works in \cite{djordjevic2023effect, li2024misalignment, elmeadawy2022hybrid, zhang2022information, raza2023data, mabena2023beam} relied on the pointing error models originally proposed for conventional FSO systems, as presented in \cite{farid2007outage}. In contrast, more recent studies \cite{dabiri2024advancing,dabiri2025oam} demonstrated that, due to the distinct beam characteristics of OAM systems compared to conventional FSO, using these traditional models led to inaccuracies and deviations from actual values. In \cite{dabiri2024advancing}, a precise model for the impact of pointing errors on intermodal crosstalk was presented, incorporating realistic channel parameters. Additionally, the study proposed detector structures to mitigate the effects of crosstalk. However, the crosstalk modeling in \cite{dabiri2024advancing} relied on solving multi-dimensional numerical integrals, which made system analysis and optimization computationally intensive.
Consequently, no closed-form analytical model of pointing errors on OAM crosstalk, incorporating realistic channel parameters, had been developed to date—a limitation that remains critical for accelerating the analysis, design, and optimization of such dynamic systems.

While existing studies have provided valuable insights, most rely on either computationally expensive Monte Carlo simulations or simplified statistical models derived from traditional FSO systems, which do not fully capture the spatial and modal characteristics unique to OAM beams. 
Pointing errors, in particular, remain one of the most critical challenges in deploying OAM-based inter-satellite communication systems, as they inherently induce intermodal crosstalk, degrading both system performance and reliability. 
The effects of these errors are complex and strongly influenced by key physical parameters such as beam divergence, tracking accuracy, and receiver aperture size. These limitations underline the need for an analytical framework that can accurately model such impairments and enable fast, online optimization of system parameters in response to the rapidly changing link conditions of dynamic LEO constellations.

\subsection{Key Contributions}
To address the lack of efficient and accurate models for analyzing pointing-induced impairments in OAM-based inter-satellite communication, this paper presents an analytical framework that precisely captures the impact of pointing errors on intermodal crosstalk.
Building upon this foundation, we further derive analytical expressions for BER performance, enabling rapid and accurate system evaluation under realistic link conditions. Unlike prior approaches, the proposed method offers a practical alternative to simulation-heavy techniques, significantly reducing computational overhead while preserving accuracy. This capability is especially important for dynamic LEO constellations, where frequent topology changes demand real-time adaptation and fast optimization. Rather than claiming to fully solve the problem, this work provides a new analytical toolset that supports deeper understanding and more effective design of OAM-enabled spaceborne optical communication systems.

This study makes the following key contributions:
\begin{itemize}
	\item \textbf{Accurate Analytical Modeling of Crosstalk:}  
	An analytical expression is developed to characterize intermodal crosstalk caused by pointing errors in OAM-based links, incorporating key physical parameters such as beam divergence, receiver aperture size, and tracking precision. This model improves upon prior approaches that rely on conventional FSO assumptions or extensive numerical simulations.
	\item \textbf{Efficient Performance Evaluation and BER Analysis:}  
	Building upon the crosstalk model, analytical expressions for BER are derived, enabling fast and reliable performance prediction. Compared to traditional Monte Carlo methods, the proposed framework reduces computation time by more than 1200$\times$ while maintaining high accuracy.
	\item \textbf{Practical Design Guidelines for OAM Links:}  
	The analysis provides actionable insights into optimal mode selection, beam width adjustment, and link configuration under varying levels of pointing accuracy and link distance. In particular, the superiority of asymmetric mode sets is demonstrated for enhancing robustness in dynamic environments.
	\item \textbf{Applicability to LEO Constellations:}  
	The proposed framework is tailored to support fast adaptation in LEO satellite constellations, where link conditions and network topology change rapidly. It enables real-time system reconfiguration, addressing a critical need in high-mobility optical satellite networks.
\end{itemize}

\begin{figure}
	\begin{center}
		\includegraphics[width=2.5 in]{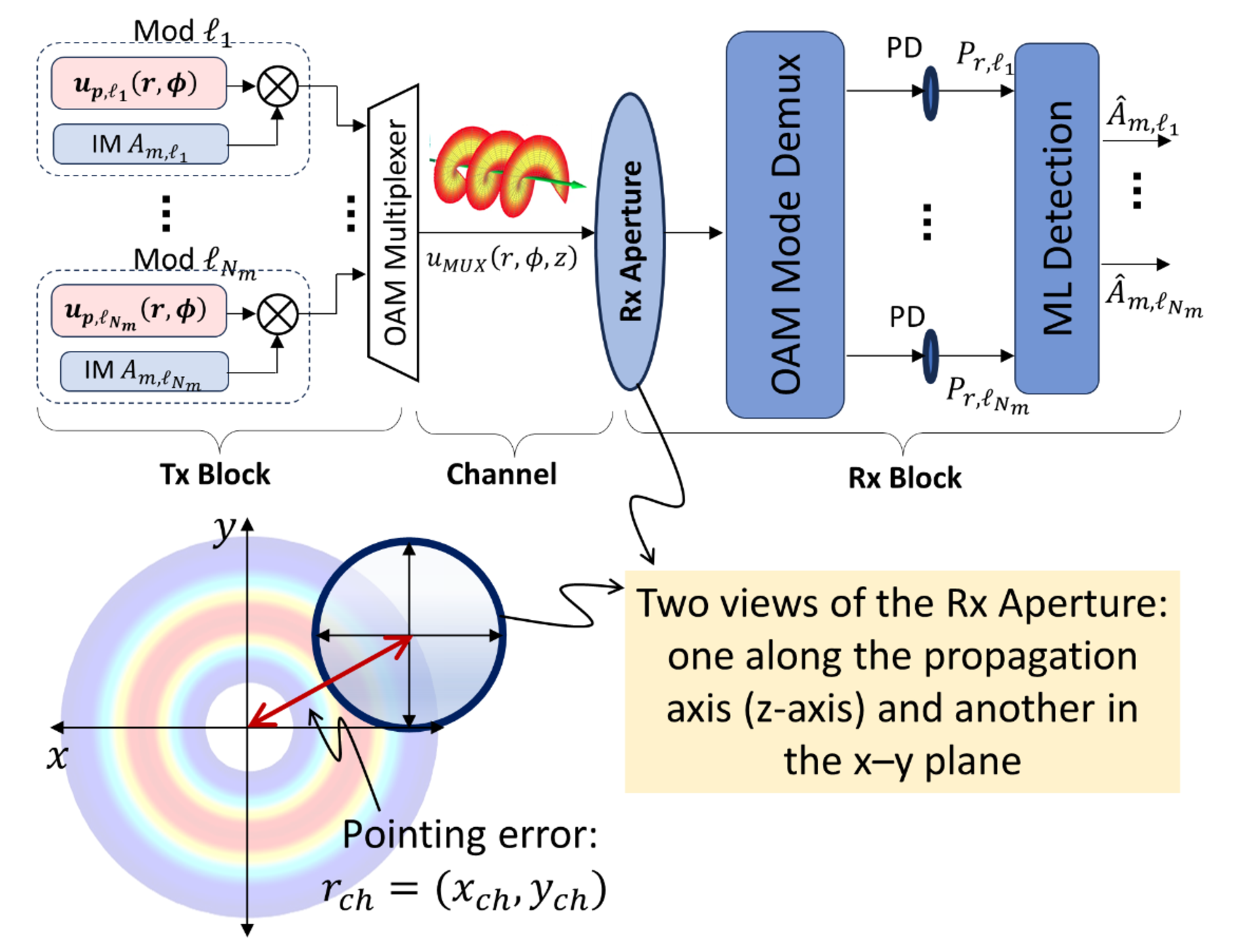}
		
		\caption{System model for an OAM-based FSO inter-satellite communication link. The transmitter employs an OAM multiplexer to combine multiple helical beams, which propagate through space and undergo pointing errors due to satellite motion and misalignment. At the receiver, the distorted beams are collected by a circular aperture, demultiplexed, and detected by an array of photodetectors (PDs). Two views of the receiver aperture are shown: one along the propagation axis (z-axis view) and another in the transverse plane (x-y view).}      
		\label{nv1}
	\end{center}
\end{figure}

\section{System Model}
As depicted in Fig. \ref{nv1}, the OAM system consists of multiple parallel channels, where each channel leverages Laguerre-Gaussian (LG) modes with unique OAM values. These modes enable the encoding of data onto separate orthogonal beams, facilitating high-capacity communication.
The characteristics of each mode are defined by the parameter $\ell$. At the transmitter, we utilize a set of orthogonal modes denoted as:
\begin{equation}
	L_m = \{\ell_1, \ell_2, \dots, \ell_{N_m}\}
\end{equation}
where $N_m$ represents the total number of modes or parallel channels, and $\ell_n$ is an integer that can take positive or negative values. For $(n, m) \in \{1, \dots, N_m\}$ with $i \neq j$, it holds that $\ell_n \neq \ell_m$. This ensures orthogonality among the modes. In fact, $\ell_n$ characterizes the unique helical phase structure of the beam, representing the OAM carried by the mode. This value uniquely defines the twist of the beam's phase around its central axis, allowing for the differentiation of parallel communication channels.

The electric field of each modulated mode in in cylindrical coordinates \( (r, \phi, z) \) is modeled as:
\begin{align}
	u_{\ell_n,m}(r, \phi, z) = \sqrt{A_{m,\ell_n}} \cdot u_{p,\ell_n}(r, \phi, z),  
\end{align}
where $A_{m,\ell_n}$ represents the independent modulation applied to mode $\ell_n$, and $u_{p,\ell_n}(r, \phi, z)$ denotes the electric field, which for a LG beam is modeled as \cite{allen1992orbital}:
\begin{align} \label{eq:LG_mode}
	& u_{\ell_n}(r, \phi, z) = \\
	&~~~\sqrt{\frac{2 p!}{\pi (p + |\ell_n|)!}} \frac{1}{w(z)} \left( \frac{\sqrt{2} r}{w(z)} \right)^{|\ell_n|} 
	L_p^{|\ell_n|} \left( \frac{2 r^2}{w(z)^2} \right)  \nonumber \\
	&\quad \times \exp\left(-\frac{r^2}{w(z)^2}\right) e^{-i \ell_n \phi} \exp\left( -i k_\nu \frac{r^2}{2 R(z)} + i  \psi(z) \right) \nonumber
\end{align}
where:
\begin{itemize}
	\item \( r \): Radial distance from the beam center.
	\item \( \phi \): Azimuthal angle in cylindrical coordinates.
	\item \( z \): Propagation distance from the transmitter.
	\item \( p \): Radial mode index (non-negative integer).
	\item \( \ell_n \): Azimuthal mode index (OAM state), which defines the orbital angular momentum of the beam.
	\item \( w(z) \): Beam radius at distance \( z \), given by
	\begin{align} \label{eq:beam_radius}
		w(z) = w_0 \sqrt{1 + \left( \frac{z}{z_R} \right)^2}
	\end{align}
	where \( w_0 \) is the beam waist radius at the focal point, and \( z_R \) is the Rayleigh range.
	\item \( L_p^{|\ell_n|} \): Associated Laguerre polynomial of order \( p \) and degree \( |\ell_n| \) which is modeled as \cite{yang2022bit}
	\begin{align}
		L_p^{|\ell_n|}(x) = \sum_{m=0}^p \frac{(-1)^m}{m!} \binom{p + |\ell_n|}{p - m} x^m
	\end{align}
	\item \( k_\nu \): Wave number, defined as \( k_\nu = \frac{2 \pi}{\lambda} \), where \( \lambda \) is the wavelength of the optical carrier.
	\item \( R(z) \): Radius of curvature of the wavefront at distance \( z \), given by
	\begin{align} \label{eq:radius_of_curvature}
		R(z) = z \left( 1 + \left( \frac{z_R}{z} \right)^2 \right)
	\end{align}
	\item \( \psi(z) \): Gouy phase shift, defined as
	\begin{align} \label{eq:gouy_phase}
		\psi(z) = (2p+|\ell_n|+1)\arctan\left(\frac{z}{z_R}\right)
	\end{align}
\end{itemize}
Finally, the transmitted field resulting from all modes is obtained by summing the contributions of each modulated mode:
\begin{align} \label{ss2}
	u_{\text{MUX}}(r, \phi, z) = \sum_{n=1}^{N_m} u_{\ell_n,m}(r, \phi, z) .
\end{align}

In inter-satellite communications, the transmitting satellite utilizes modern tracking systems to direct signals toward the receiving satellite. Let $\theta_t = (\theta_{tx}, \theta_{ty})$ represent the tracking system errors in the $x$ and $y$ directions, where $\theta_{tx} \sim \mathcal{N}(0, \sigma_\theta^2)$ and $\theta_{ty} \sim \mathcal{N}(0, \sigma_\theta^2)$ \cite{dabiri2019optimal}. Small deviations in the transmitter's tracking system, $\theta_t = (\theta_{tx}, \theta_{ty})$, result in the displacement of the beam center relative to the center of the receiver's aperture by \cite{dabiri2019tractable}:
\begin{align}
	\label{sd2}
	r_{ch} = \sqrt{x_{ch}^2 + y_{ch}^2}, ~~\text{where}~~
	\begin{cases}
		&\!
		\!\! x_{ch} = \theta_{tx} Z, \\
		&\!
		\!\! y_{ch} = \theta_{ty} Z,
	\end{cases}
\end{align}
and $Z$ is the link length.

At the receiver, the signal is collected by a circular aperture with radius $r_a$. Let $(r', \phi', z')$ and $(x', y', z')$ denote the cylindrical and Cartesian coordinate systems from the receiver's perspective, respectively, where the propagation axis is aligned such that $z = z'$. Following the procedures outlined in \cite{dabiri2024advancing}, the electric field of the OAM signal, misaligned by $r_{ch} = (x_{ch}, y_{ch})$, and received by the aperture is modeled in the Cartesian coordinate system as:
\begin{align} \label{fiel_prim1}
	u_{\text{ap}} &(x', y', z')   = \sum_{n=1}^{N_m} \sqrt{A_{m,\ell_n}} u_{\text{ap},\ell_n}(x', y', z'),
\end{align}
where
\begin{align} \label{fiel_prim}
	u_{\text{ap},\ell_n}&(x', y', z')   =  \sqrt{\frac{2 p!}{\pi (p + |\ell_n|)!}} \frac{1}{w(z')}
	e^{-i \ell_n \tan^{-1} \left( \frac{y' + y_{\text{ch}}}{x' + x_{\text{ch}}} \right)} \nonumber \\
	& \times \left( \frac{\sqrt{2} \sqrt{(x' + x_{\text{ch}})^2 + (y' + y_{\text{ch}})^2}}{w(z')} \right)^{|\ell_n|} \nonumber \\
	&\times L_p^{|\ell_n|} \left( \frac{2 ((x' + x_{\text{ch}})^2 + (y' + y_{\text{ch}})^2)}{w(z')^2} \right) \nonumber \\
	&\times \exp\left(-\frac{(x' + x_{\text{ch}})^2 + (y' + y_{\text{ch}})^2}{w(z')^2}\right) \nonumber \\
	&\times  \exp\left( -i k_\nu \frac{(x' + x_{\text{ch}})^2 + (y' + y_{\text{ch}})^2}{2 R(z')} + i  \psi(z') \right),
\end{align}
where $ -\sqrt{r_a^2-x'^2}<y'<\sqrt{r_a^2-x'^2}$, and $-r_z<x'<r_a$.
Given that \( x' = r' \cos(\phi') \) and \( y' = r' \sin(\phi') \), we can rewrite \eqref{fiel_prim} as \eqref{fiel_prim2}.
\begin{figure*}[!t]
	\normalsize
	\begin{align} \label{fiel_prim2}
		&u_{\text{ap},\ell_n}(r', \phi', z)  =  \sqrt{\frac{2 p!}{\pi (p + |\ell_n|)!}} \frac{1}{w(z')} \left( \frac{\sqrt{2} \sqrt{(r' \cos(\phi') + x_{\text{ch}})^2 + (r' \sin(\phi') + y_{\text{ch}})^2}}{w(z')} \right)^{|\ell_n|}  \\
		&\times L_p^{|\ell_n|} \left( \frac{2 \left( (r' \cos(\phi') + x_{\text{ch}})^2 + (r' \sin(\phi') + y_{\text{ch}})^2 \right)}{w(z')^2} \right) 
		\exp\left(-\frac{(r' \cos(\phi') + x_{\text{ch}})^2 + (r' \sin(\phi') + y_{\text{ch}})^2}{w(z')^2}\right) \nonumber \\
		&\times e^{-i \ell_n \tan^{-1} \left( \frac{r' \sin(\phi') + y_{\text{ch}}}{r' \cos(\phi') + x_{\text{ch}}} \right)} \exp\left( -i k_\nu \frac{(r' \cos(\phi') + x_{\text{ch}})^2 + (r' \sin(\phi') + y_{\text{ch}})^2}{2 R(z')} + i  \psi(z') \right),~~~~~\text{for}~ r'<r_a,~~0<\phi'<2\pi.  \nonumber 
	\end{align}
	\hrulefill
\end{figure*}    
The electric field at the aperture output, $u_{\text{ap}}(r', \phi', z)$, is divided into $N_m$ segments, each fed into the input of a phase filter. Each filter multiplies the input signal by $e^{i\phi' \ell_j'}$ and, through integration over $\phi'$. The electric field at the filter output is modeled as:
\begin{align} \label{f1}
	S_{\ell_j'}(r', z) = \sum_{n=1}^{N_m}\int_0^{2\pi} \frac{ \sqrt{A_{m,\ell_n}} u_{\text{ap},\ell_n}(r', \phi', z)}{\sqrt{2\pi}N_m} \, e^{i \ell_j' \phi'} \, d\phi',
\end{align}
where $j\in\{1,...,N_m\}$.
According to \eqref{f1}, under ideal conditions and without pointing errors, only the electrical field corresponding to the mode $\ell_n = \ell'_j$ can pass through the filter, while all other modes are completely filtered out. However, as shown in \eqref{fiel_prim2}, pointing errors distort the electric field, preventing the filter from fully rejecting modes $\ell_n \neq \ell'_j$. In this case, we have:
\begin{align} \label{f2}
	S_{\ell_j'}(r', z) = \sum_{n=1}^{N_m} \sqrt{A_{m,\ell_n}} S_{\ell_n,\ell_j'}(r', z),
\end{align}
where
\begin{align} \label{f3}
	S_{\ell_n,\ell_j'}(r', z) = \int_0^{2\pi} \frac{  u_{\text{ap},\ell_n}(r', \phi', z)}{\sqrt{2\pi}N_m} \, e^{i \ell_j' \phi'} \, d\phi'.
\end{align}

After filtering, each LG mode component $S_{\ell_j'}(r', z)$ is directed to an Avalanche Photodiode (APD), which converts the optical signal into an electrical current with a multiplication gain \( G \). The received power for each mode \( \ell_j' \) at the APD is obtained as:
\begin{align} \label{received_power}
	P_{r,\ell'} =   \sum_{n=1}^{N_m} A_{m,\ell_n} C_{\ell_n,\ell'_j},
\end{align}
where \( \eta \) is the photodetector’s responsivity, and $C_{\ell_n,\ell'_j}$ represents the cross-talk, modeled as:
\begin{align} \label{f4}
	C_{\ell_n,\ell'_j} = 2\pi \eta G \int_0^{r_a}  \left|   S_{\ell_n,\ell_j'}(r', z) \right|^2 \, r'  \, dr'. 
\end{align}

\section{Cross-Talk Analysis}
This section investigates the impact of pointing errors on cross-talk characteristics in OAM-based inter-satellite links. 

\subsection{Cross-Talk Modeling and Analysis}
Substituting \eqref{fiel_prim1}, \eqref{fiel_prim2}, and \eqref{f3} in \eqref{f4}, the final form of the cross-talk can be expressed as \eqref{cros2}. As observed in \eqref{cros2}, the cross-talk is a highly complex and nonlinear function of various channel parameters, such as the pointing error \( r_{ch} = (x_{ch}, y_{ch}) \), beam width, mode index \( \ell_n \), link length, and phase. In this paper, we demonstrate that even small variations in \( r_{ch} = (x_{ch}, y_{ch}) \) result in significant changes in cross-talk.
\begin{figure*}[!t]
	\normalsize
	\begin{align} \label{cros2}
		&C_{\ell_n,\ell'_j}  =  \frac{  \eta G }{N_m^2 w^2(z')}  \int_0^{r_a}  \Bigg| 
		\int_0^{2\pi}   
		\sqrt{\frac{2 p!}{\pi (p + |\ell_n|)!}}  \left( \frac{\sqrt{2} \sqrt{(r' \cos(\phi') + x_{\text{ch}})^2 + (r' \sin(\phi') + y_{\text{ch}})^2}}{w(z')} \right)^{|\ell_n|}  \\
		&\times L_p^{|\ell_n|} \left( \frac{2 \left( (r' \cos(\phi') + x_{\text{ch}})^2 + (r' \sin(\phi') + y_{\text{ch}})^2 \right)}{w(z')^2} \right) 
		\exp\left(-\frac{(r' \cos(\phi') + x_{\text{ch}})^2 + (r' \sin(\phi') + y_{\text{ch}})^2}{w(z')^2}\right) \nonumber \\
		&\times \exp\left({-i \ell_n \tan^{-1} \left( \frac{r' \sin(\phi') + y_{\text{ch}}}{r' \cos(\phi') + x_{\text{ch}}} \right)
			+ i \ell_j' \phi'} \right) 
		\exp\left( -i k_\nu \frac{(r' \cos(\phi') + x_{\text{ch}})^2 + (r' \sin(\phi') + y_{\text{ch}})^2}{2 R(z')} + i  \psi(z') \right)
		\, d\phi'
		\Bigg|^2\, r'  \, dr'  \nonumber 
	\end{align}
	\hrulefill
\end{figure*}

However, similar to most works in the literature, the evaluation of cross-talk involves solving a two-dimensional numerical integral. This requirement significantly increases the computational complexity, particularly in the design and optimization of OAM-based systems. 
For instance, to perform detection using a simulation-based approach in \cite{dabiri2024advancing}, it is necessary to generate a large number of random bits and pointing error coefficients. For each coefficient, the two-dimensional integral in \eqref{cros2} must be solved \( N_m^2 \) times, corresponding to \( N_m \) different values of \( \ell_n \) and \( N_m \) different values of \( \ell'_j \).
Given this challenge, the main goal of this section is to derive less computationally intensive yet accurate formulations to facilitate the analysis and optimal design of OAM-based systems.
In the following, several propositions are presented to derive less computationally intensive expressions for evaluating \eqref{cros2}.

\begin{proposition}
	\( C_{\ell_n,\ell'_j} \) can be calculated with lower computational complexity compared to \eqref{cros2}, using only a one-dimensional integral over \( \phi' \) as follows:
\begin{align} \label{cros4}
	&C_{\ell_n,\ell'_j}  \simeq  \frac{A_1 r_a}{K_r}  \sum_{k=1}^{K_r}  \Bigg| 
	\int_0^{2\pi}   
	\left(  \sqrt{ 2 A_2\left(r_a k/K_r\right) } \right)^{|\ell_n|}   \nonumber \\
	&\times L_p^{|\ell_n|} \left( 2 A_2\left(\frac{r_a k}{K_r}\right) \right) 
	e^{-A_2\left(\frac{r_a k}{K_r}\right)}  
	e^{-i \ell_n \tan^{-1} \left( A_3\left(\frac{r_a k}{K_r}\right) \right)
		+ i \ell_j' \phi'}\nonumber \\
	&\times   
	e^{ -i k_\nu A_4\left(\frac{r_a k}{K_r}\right) + i  \psi(z') } 
	\, d\phi'
	\Bigg|^2   ,
\end{align}
where
\begin{align}
	\begin{cases}
		&\!\!\! A_1 = \frac{  \eta G }{N_m^2 w^2(z')}  \frac{2 p!}{\pi (p + |\ell_n|)!},\\
		&\!\!\! A_2(r'_k) = \frac{(r'_k \cos(\phi') + x_{\text{ch}})^2 + (r'_k \sin(\phi') + y_{\text{ch}})^2}{w^2(z')} \\
		&\!\!\! A_3(r'_k) = \frac{r'_k \sin(\phi') + y_{\text{ch}}}{r'_k \cos(\phi') + x_{\text{ch}}} \\
		&\!\!\! A_4(r'_k) = \frac{(r'_k \cos(\phi') + x_{\text{ch}})^2 + (r'_k \sin(\phi') + y_{\text{ch}})^2}{2 R(z')}
	\end{cases}
\end{align}
and $ r'_k = \left(\frac{r_a k}{K_r}\right) $.
\end{proposition}

\begin{proof} Please refer to Appendix \ref{AppA}.
\end{proof}

\eqref{cros4} in Proposition 1 replaces the two-dimensional integral with \( K_r \) evaluations of a one-dimensional integral for calculating \( C_{\ell_n,\ell'_j} \). As demonstrated in the simulation section, using only \( K_r = 6 \) yields results very close to the exact solution obtained from the two-dimensional integral, while reducing the computational complexity and runtime by more than 10 times.

\begin{proposition}
	\( C_{\ell_n,\ell'_j} \) can be calculated with lower computational complexity compared to \eqref{cros2}, using only a one-dimensional integral over \( r' \) as follows:
	\begin{align} \label{cros8}
		&C_{\ell_n,\ell'_j}  \simeq  A_1 A_5  
		\int_0^{r_a}   r'  \left( J_{\ell_j'}\left( \frac{k  r_{\text{ch}}}{R(z')} r' \right)   \right)^2     \, dr'  
	\end{align}
	where 
	\begin{align}
		A_5 = \left[ 2\pi \left( \frac{\sqrt{2} r_{\text{ch}} }{w(z')} \right)^{|\ell_n|}  
		L_p^{|\ell_n|} \left( \frac{2 r^2_{\text{ch}} }{w(z')^2} \right) 
		\exp\left(-\frac{r_{\text{ch}}^2}{w(z')^2}\right)  \right]^2. \nonumber
	\end{align}
\end{proposition}


\begin{proof} Please refer to Appendix \ref{AppB}.

\end{proof}

It should be noted that \eqref{cros8}, presented in Proposition 2, has a lower computational burden compared to the expression in Proposition 1. Simulation results also demonstrate that it provides high accuracy for a wide range of \( r_{\text{ch}} \) values. Only slight deviations are observed for small \( r_{\text{ch}} \) values, which are highly unlikely in practice for inter-satellite communications due to the long link distances. Furthermore, for such small \( r_{\text{ch}} \) values, the system performance remains very robust.

\begin{proposition}
	The closed-form expression for \( C_{\ell_n,\ell'_j} \) is derived as follows:
	\begin{align} \label{cros9}
		&C_{\ell_n,\ell'_j}  \simeq  A_1 A_5  \frac{r_a^2 }{K_r^2}
		\sum_{k=1}^{K_r}    k \left( J_{\ell_j'}\left( \frac{k_\nu  r_{\text{ch}} r_a}{R(z') K_r} k \right)   \right)^2 
	\end{align}
	
\end{proposition}

\begin{proof} Please refer to Appendix \ref{AppC}.
	
\end{proof}

For \eqref{cros9} in Proposition 3, it can be shown that for a limited number of \( K_r \) (between 5 and 8), \( C_{\ell_n,\ell'_j} \) can be calculated with high accuracy compared to the two-dimensional integral form in \eqref{cros2}. Moreover, \eqref{cros9} has a simpler and more tractable form, reducing computational complexity by more than 100 times compared to \eqref{cros2}. This significant reduction enables faster analysis and optimal design of OAM-based inter-satellite communication systems.

Based on Proposition 3, several key insights and properties related to \( C_{\ell_n,\ell'_j} \) for inter-satellite communications under pointing errors can be highlighted, which can lead to a better understanding and, consequently, the optimal design of such systems.

\begin{proposition}
	For large values of $r_{\text{ch}}$, the asymptotic expression for \( C_{\ell_n,\ell'_j} \) is derived as follows:
	\begin{align} \label{cros10}
		&C_{\ell_n,\ell'_j}  \simeq  A_1 A_5   
		       \frac{R(z') r_a }{\pi k_\nu  r_{\text{ch}} }    
	\end{align}	
\end{proposition}

\begin{proof} Please refer to Appendix \ref{AppD}.
	
\end{proof}

Proposition 4 effectively characterizes the behavior of an OAM-based system under severe pointing errors. From Proposition 4, the following key result can be derived.

\textbf{Remark 1:} For large values of \(r_{\text{ch}}\), the crosstalk function \(C_{\ell_n, \ell'_j}\) becomes independent of the OAM mode indices \(\ell'_j\). 

OAM modes carry two types of information: radial and phase-based. Remark 1 indicates that, for large values of \(r_{\text{ch}}\), the phase-based information is entirely lost. At the receiver, for different values of \(\ell'_j\), the crosstalk terms \(C_{\ell_n, \ell'_j}\) tend to converge to similar values. However, it should be noted that the parameters \(\ell_n\) and \(\ell'_j\) are distinct. The term \(A_5\) in \eqref{cros10} remains a function of \(\ell_n\), and its variation leads to different outputs for \(C_{\ell_n, \ell'_j}\). These variations are attributed to differences in the radial structure of the modes.

\textbf{Remark 2:} For inter-satellite communications, it can be observed with high accuracy that:
\begin{align}
C_{\ell_n, \ell'_j} \simeq C_{-\ell_n, \ell'_j}.
\end{align}

This property in Remark 2 is clearly observable from the relations in Propositions 3 and 4. It should be noted that this characteristic applies specifically to long-distance inter-satellite communications, where the beam width becomes significantly larger than the receiver aperture due to the extended link distance. This is in contrast to short-range terrestrial OAM communications, where the beam width is much smaller and comparable to the receiver aperture.


\subsection{BER Analysis}
In this subsection, the BER performance of the OAM-based inter-satellite link is analytically evaluated. The received signal corresponding to the $j$-th OAM mode, after phase filtering and optical-to-electrical conversion, can be modeled as
\begin{align} \label{eq:rx_signal}
	y_j = \sum_{n=1}^{N_m} \sqrt{A_{m,\ell_n}}\, \sqrt{C_{\ell_n,\ell'_j}} + n_j,
\end{align}
where $n_j$ denotes the additive white Gaussian noise (AWGN) term with distribution $n_j \sim \mathcal{CN}(0,N_0)$, and $C_{\ell_n,\ell'_j}$ is the crosstalk coefficient defined previously.
The receiver collects the outputs from all $N_m$ filters, forming the received signal vector:
\begin{align} \label{eq:rx_vector}
	\mathbf{y} = [y_1, y_2, \dots, y_{N_m}]^T.
\end{align}
Each element $y_j$ corresponds to the output of the $j$-th OAM mode after phase filtering and optical-to-electrical conversion.
Given the transmitted data vector $\mathbf{A}_m = [A_{m,\ell_1}, A_{m,\ell_2}, \dots, A_{m,\ell_{N_m}}]^T$, the maximum likelihood (ML) detection rule jointly estimates the transmitted symbols by solving:
\begin{align} \label{eq:ML_joint}
	\hat{\mathbf{A}}_m = \arg\min_{\mathbf{A}_m \in \mathcal{A}^{N_m}} \left\| \mathbf{y} - \mathbf{H}(r_{\text{ch}})\, \sqrt{\mathbf{A}_m} \right\|^2,
\end{align}
where $\mathcal{A}$ is the modulation alphabet (e.g., $\mathcal{A} = \{0,1\}$ for OOK), the square root is applied element-wise, and $\mathbf{H}(r_{\text{ch}})$ is the effective channel gain matrix conditioned on the instantaneous pointing error $r_{\text{ch}}$, defined as
\begin{align} \label{eq:H_matrix}
	\mathbf{H}(r_{\text{ch}}) = 
	\begin{bmatrix}
		\sqrt{C_{\ell_1,\ell'_1}} & \sqrt{C_{\ell_2,\ell'_1}} & \cdots & \sqrt{C_{\ell_{N_m},\ell'_1}} \\
		\sqrt{C_{\ell_1,\ell'_2}} & \sqrt{C_{\ell_2,\ell'_2}} & \cdots & \sqrt{C_{\ell_{N_m},\ell'_2}} \\
		\vdots & \vdots & \ddots & \vdots \\
		\sqrt{C_{\ell_1,\ell'_{N_m}}} & \sqrt{C_{\ell_2,\ell'_{N_m}}} & \cdots & \sqrt{C_{\ell_{N_m},\ell'_{N_m}}}
	\end{bmatrix}.
\end{align}
It is noted that each element $C_{\ell_n,\ell'_j}$ implicitly depends on the instantaneous pointing error $r_{\text{ch}}$ as characterized in \eqref{cros9}.

\subsection{Error Probability Analysis for $N_m=2$ and OOK}
The unconditional average probability of error is obtained by averaging \( P_{\text{e,avg}|r_{\text{ch}}} \) over the probability density function (PDF) of the pointing error displacement \( r_{\text{ch}} \), denoted by \( f_{r_{\text{ch}}}(r) \). Thus, the overall average BER is expressed as
\begin{align} \label{eq:ber_total_unconditional}
	P_{\text{e,avg}} = \int_0^{\infty} P_{\text{e,avg}|r_{\text{ch}}}\, f_{r_{\text{ch}}}(r_{\text{ch}}) \, dr_{\text{ch}}.
\end{align}
Due to the high sensitivity of OAM-based inter-satellite links to pointing errors and intermodal crosstalk, it is often preferable to restrict the number of simultaneously transmitted modes to $N_m=2$ in practical deployments. Higher values of $N_m$ can be considered only when the pointing error is sufficiently small, ensuring acceptable crosstalk levels.
Assuming $N_m=2$ and OOK modulation ($A_{m,\ell_j} \in \{0,1\}$), at each symbol interval, the transmitted vector can be one of the following four possible cases:
\begin{align}
	\{[0,0]^T, [1,0]^T, [0,1]^T, [1,1]^T\}.
\end{align} 
Given the four possible transmitted symbol vectors, the conditional bit error probability depends on the specific transmitted vector.
Assuming that each transmitted vector \(\mathbf{A}_m \in \{ [0,0]^T, [1,0]^T, [0,1]^T, [1,1]^T \}\) occurs with equal probability, the average probability of error conditioned on \(r_{\text{ch}}\) is given by
\begin{align} \label{eq:ber_total_expression}
	P_{\text{e,avg}|r_{\text{ch}}} = \frac{1}{4} \left( P_{\text{e},1|r_{\text{ch}}} + P_{\text{e},2|r_{\text{ch}}} + P_{\text{e},3|r_{\text{ch}}} + P_{\text{e},4|r_{\text{ch}}} \right),
\end{align}
where \(P_{\text{e},1|r_{\text{ch}}}\), \(P_{\text{e},2|r_{\text{ch}}}\), \(P_{\text{e},3|r_{\text{ch}}}\), and \(P_{\text{e},4|r_{\text{ch}}}\) are derived in the following.

\subsubsection{Case 1: Transmission of \([0,0]^T\)}
When the transmitted vector is \(\mathbf{A}_m = [0,0]^T\), the received signal consists purely of real additive Gaussian noise:
\begin{align} \label{eq:case1_received_real}
	\mathbf{y} = \mathbf{n},
\end{align}
where \(\mathbf{n} \sim \mathcal{N}(\mathbf{0},N_0\mathbf{I}_2)\) is a real-valued Gaussian noise vector.
Following the ML detection rule, an error occurs if the metric corresponding to any incorrect transmitted vector has a smaller value than the metric of the correct vector. Therefore, the probability of error is
\begin{align} \label{eq:case1_error_condition_real}
	P_{\text{e},1|r_{\text{ch}}} = \Pr\left( \min_{\mathbf{A}_m \neq [0,0]^T} \left\| \mathbf{y} - \mathbf{H}(r_{\text{ch}}) \sqrt{\mathbf{A}_m} \right\|^2 < \|\mathbf{y}\|^2 \right).
\end{align}
In particular, an error occurs if at least one of the following events happens:
\begin{itemize}
	\item The metric of \([1,0]^T\) is smaller than that of \([0,0]^T\).
	\item The metric of \([0,1]^T\) is smaller than that of \([0,0]^T\).
	\item The metric of \([1,1]^T\) is smaller than that of \([0,0]^T\).
\end{itemize}
Considering the comparison with \([1,0]^T\), the condition for an error can be simplified as
\begin{align} \label{eq:case1_metric_compare_real}
	2 \mathbf{h}_1^T \mathbf{n} > \|\mathbf{h}_1\|^2,
\end{align}
where \(\mathbf{h}_1 = \left[ \sqrt{C_{\ell_1,\ell'_1}}, \sqrt{C_{\ell_1,\ell'_2}} \right]^T\).
Since \(\mathbf{n}\) is a real Gaussian vector, the random variable \(2 \mathbf{h}_1^T \mathbf{n}\) follows a real Gaussian distribution with zero mean and variance \(4N_0 \|\mathbf{h}_1\|^2\).
Thus, the probability of erroneously favoring \([1,0]^T\) over \([0,0]^T\) is given by
\begin{align} \label{eq:case1_error_10_real}
	P_{\text{e},1\to2|r_{\text{ch}}} = Q\left( \sqrt{ \frac{ \|\mathbf{h}_1\|^2 }{4N_0} } \right).
\end{align}

Similarly, the error probability for incorrectly favoring \([0,1]^T\) is
\begin{align} \label{eq:case1_error_01_real}
	P_{\text{e},1\to3|r_{\text{ch}}} = Q\left( \sqrt{ \frac{ \|\mathbf{h}_2\|^2 }{4N_0} } \right),
\end{align}
where \(\mathbf{h}_2 = \left[ \sqrt{C_{\ell_2,\ell'_1}}, \sqrt{C_{\ell_2,\ell'_2}} \right]^T\).
For the case of incorrectly favoring \([1,1]^T\), the error probability is
\begin{align} \label{eq:case1_error_11_real}
	P_{\text{e},1\to4|r_{\text{ch}}} = Q\left( \sqrt{ \frac{ \|\mathbf{h}_1+\mathbf{h}_2\|^2 }{4N_0} } \right).
\end{align}
Therefore, the overall probability of error for Case 1, conditioned on \(r_{\text{ch}}\), can be obtained by using \eqref{eq:case1_error_10_real}, \eqref{eq:case1_error_01_real}, and \eqref{eq:case1_error_11_real}, and is explicitly given by
\begin{align} \label{eq:case1_total_error_real_final}
	&P_{\text{e},1|r_{\text{ch}}} =  \\
	&Q\left( \sqrt{ \frac{ \|\mathbf{h}_1\|^2 }{4N_0} } \right) + Q\left( \sqrt{ \frac{ \|\mathbf{h}_2\|^2 }{4N_0} } \right) + Q\left( \sqrt{ \frac{ \|\mathbf{h}_1+\mathbf{h}_2\|^2 }{4N_0} } \right), \nonumber
\end{align}
where \(\mathbf{h}_1 = \left[ \sqrt{C_{\ell_1,\ell'_1}}, \sqrt{C_{\ell_1,\ell'_2}} \right]^T\) and \(\mathbf{h}_2 = \left[ \sqrt{C_{\ell_2,\ell'_1}}, \sqrt{C_{\ell_2,\ell'_2}} \right]^T\).
\subsubsection{Cases 2--4: Transmission of \([1,0]^T\), \([0,1]^T\), and \([1,1]^T\)}
When the transmitted vector is either \(\mathbf{A}_m = [1,0]^T\), \(\mathbf{A}_m = [0,1]^T\), or \(\mathbf{A}_m = [1,1]^T\), the received signal can be expressed as
\begin{align}
	\mathbf{y} = 
	\begin{cases}
		\mathbf{h}_1 + \mathbf{n}, & \text{for } [1,0]^T, \\
		\mathbf{h}_2 + \mathbf{n}, & \text{for } [0,1]^T, \\
		\mathbf{h}_1 + \mathbf{h}_2 + \mathbf{n}, & \text{for } [1,1]^T.
	\end{cases}
\end{align}
Following the same ML detection procedure as in Case 1, the probability of error for Cases 2 and 3 (transmission of \([1,0]^T\) and \([0,1]^T\)) can be obtained as
\begin{align} \label{eq:case23_total_error_real_final}
	& P_{\text{e},2|r_{\text{ch}}} = P_{\text{e},3|r_{\text{ch}}} = \\
	&Q\left( \sqrt{ \frac{ \|\mathbf{h}_1\|^2 }{4N_0} } \right) + Q\left( \sqrt{ \frac{ \|\mathbf{h}_2\|^2 }{4N_0} } \right) + Q\left( \sqrt{ \frac{ \|\mathbf{h}_1-\mathbf{h}_2\|^2 }{4N_0} } \right). \nonumber
\end{align}

For Case 4 (transmission of \([1,1]^T\)), it can be observed that the error probability is identical to that of Case 1, thus
\begin{align} \label{eq:case4_total_error_real_final}
	P_{\text{e},4|r_{\text{ch}}} = P_{\text{e},1|r_{\text{ch}}},
\end{align}
where \(P_{\text{e},1|r_{\text{ch}}}\) is given in \eqref{eq:case1_total_error_real_final}.
%

Finally, by substituting the expressions from \eqref{eq:case1_total_error_real_final}--\eqref{eq:case4_total_error_real_final}, and noting that \(P_{\text{e},1|r_{\text{ch}}} = P_{\text{e},4|r_{\text{ch}}}\) and \(P_{\text{e},2|r_{\text{ch}}} = P_{\text{e},3|r_{\text{ch}}}\), the average probability of error simplifies to
\begin{align} \label{eq:ber_total_final_correct}
	&P_{\text{e,avg}|r_{\text{ch}}} = Q\left( \sqrt{ \frac{ \|\mathbf{h}_1\|^2 }{4N_0} } \right) + Q\left( \sqrt{ \frac{ \|\mathbf{h}_2\|^2 }{4N_0} } \right) \\
	&~~~+ \frac{1}{2} Q\left( \sqrt{ \frac{ \|\mathbf{h}_1+\mathbf{h}_2\|^2 }{4N_0} } \right) + \frac{1}{2} Q\left( \sqrt{ \frac{ \|\mathbf{h}_1-\mathbf{h}_2\|^2 }{4N_0} } \right). \nonumber
\end{align}

\begin{table} 
	\caption{Default parameter values used in the simulations.} 
	\centering 
	\begin{tabular}{c } 
		\includegraphics[width=3 in]{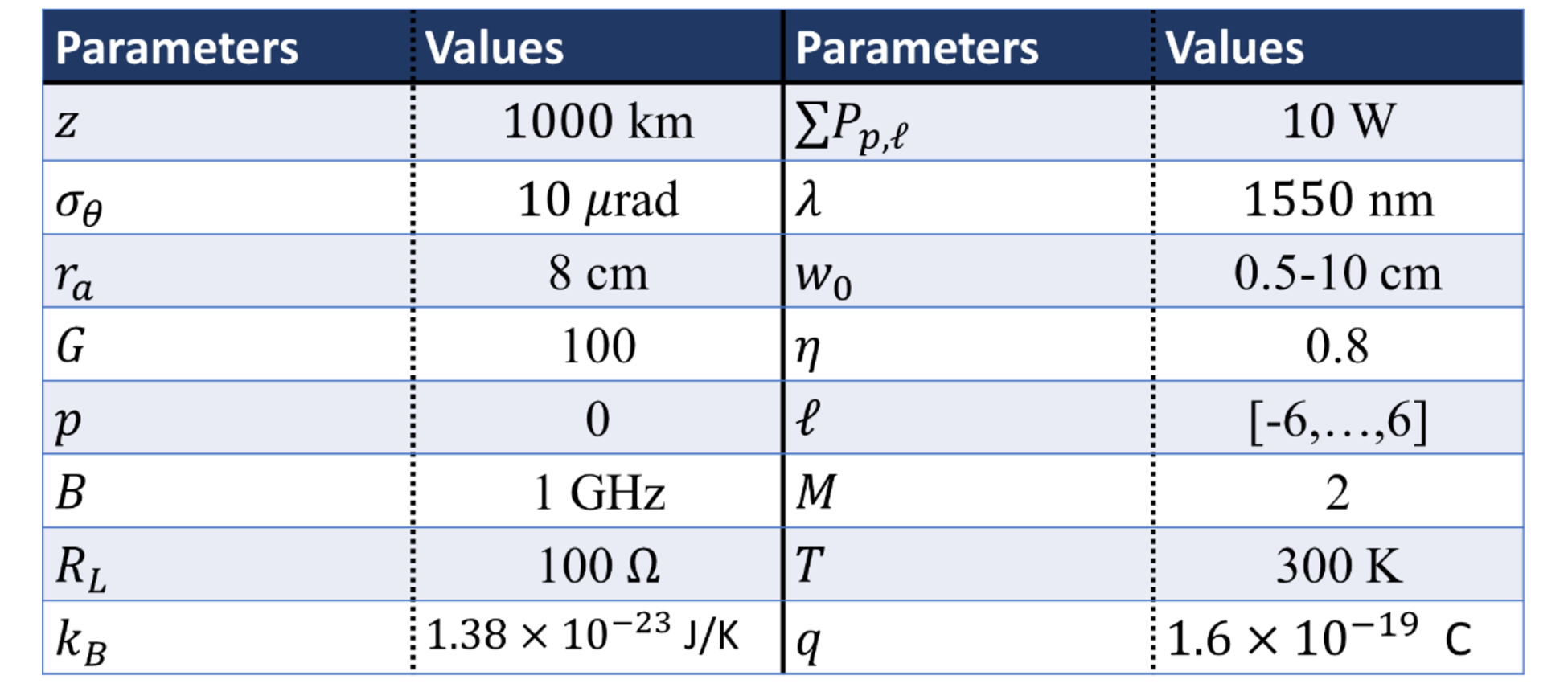}     
	\end{tabular}
	\label{Tab1} 
\end{table}

\begin{figure*}
	\centering
	\subfloat[] {\includegraphics[width=2.3 in, height=1.4 in]{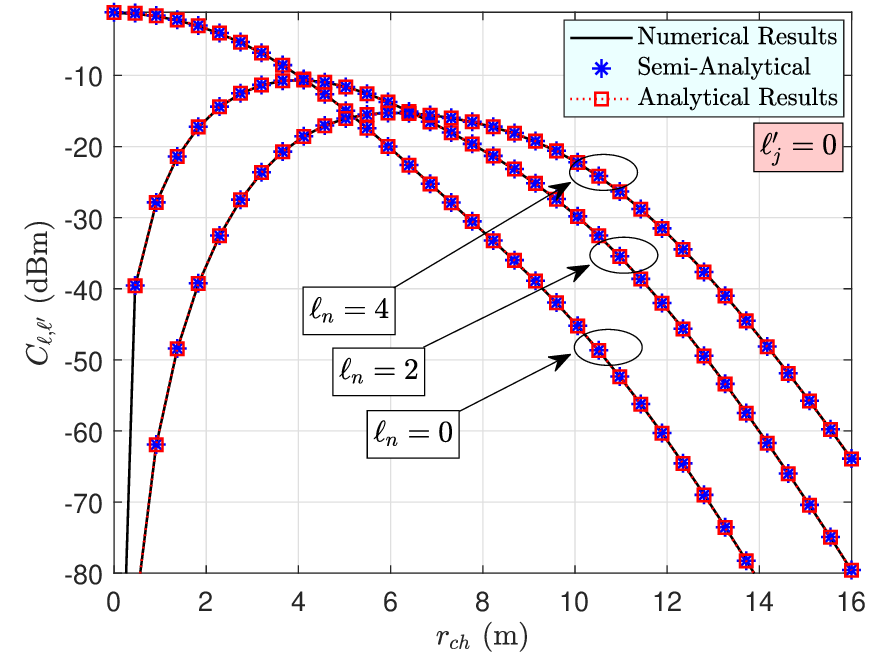}
		\label{cf1}
	}
	\hfill
	\subfloat[] {\includegraphics[width=2.3 in, height=1.4 in]{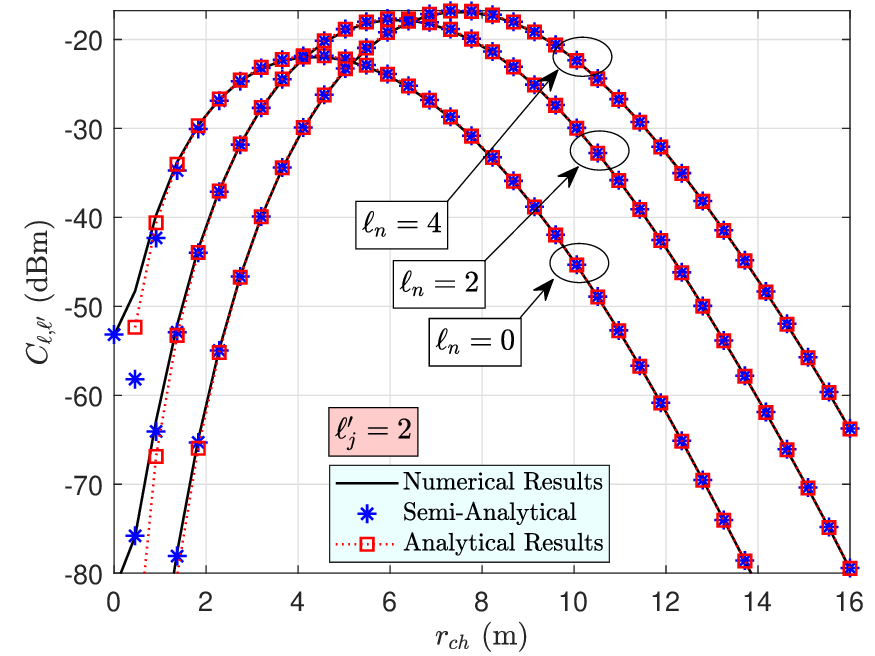}
		\label{cf2}
	}
	\hfill
	\subfloat[] {\includegraphics[width=2.3 in, height=1.4 in]{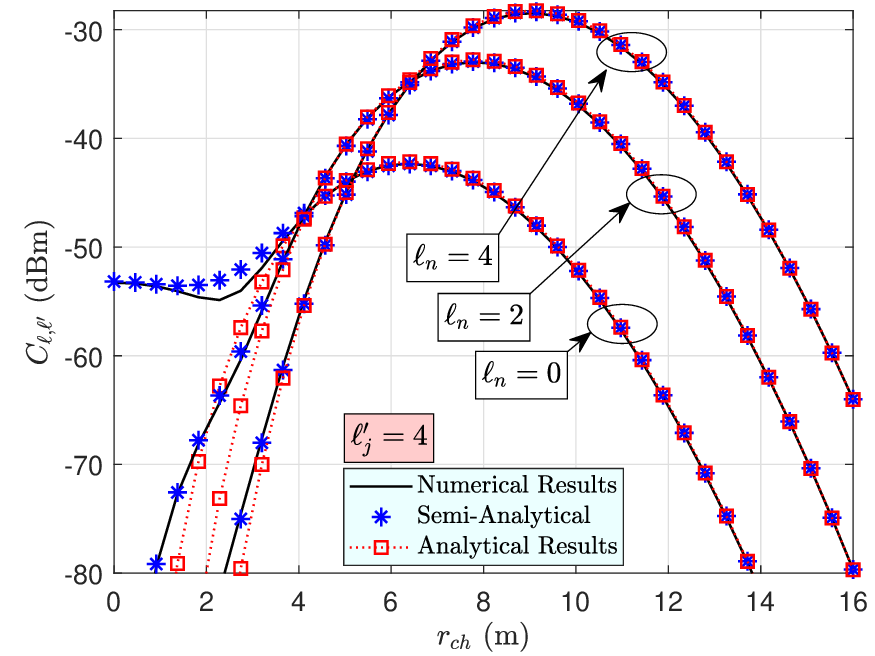}
		\label{cf3}
	}
	\caption{The crosstalk \( C_{\ell_n, \ell'_j} \) versus the pointing error \( r_{\text{ch}} \) for $p=0$, $w_0=5$ cm, and different values of \( (\ell_n, \ell'_j) \):  
		(a) \( \ell'_j = 0 \),  
		(b) \( \ell'_j = 2 \), and  
		(c) \( \ell'_j = 4 \).}
	\label{cf}
\end{figure*}

\begin{figure*}
	\centering
	\subfloat[] {\includegraphics[width=2.3 in, height=1.4 in]{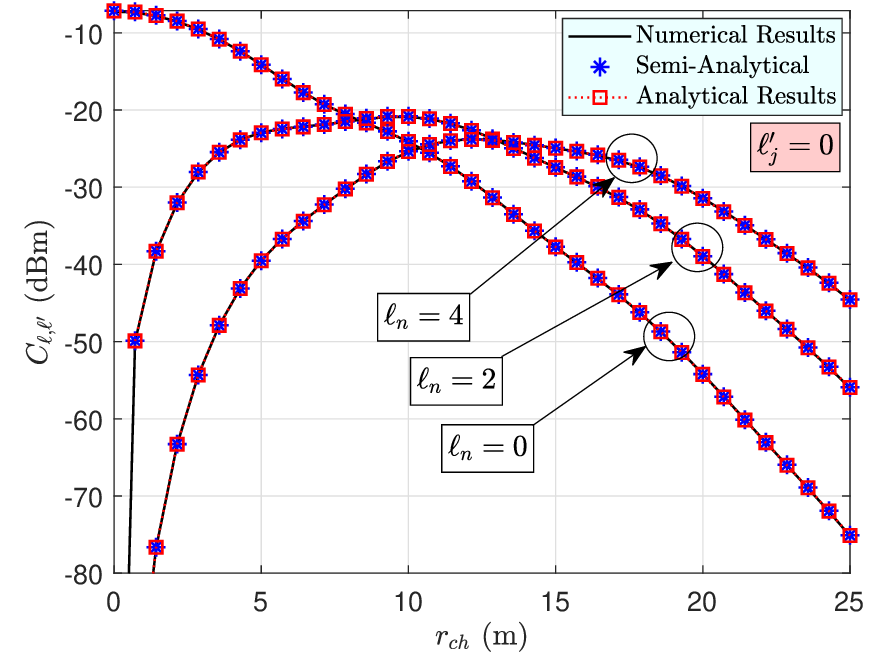}
		\label{cb1}
	}
	\hfill
	\subfloat[] {\includegraphics[width=2.3 in, height=1.4 in]{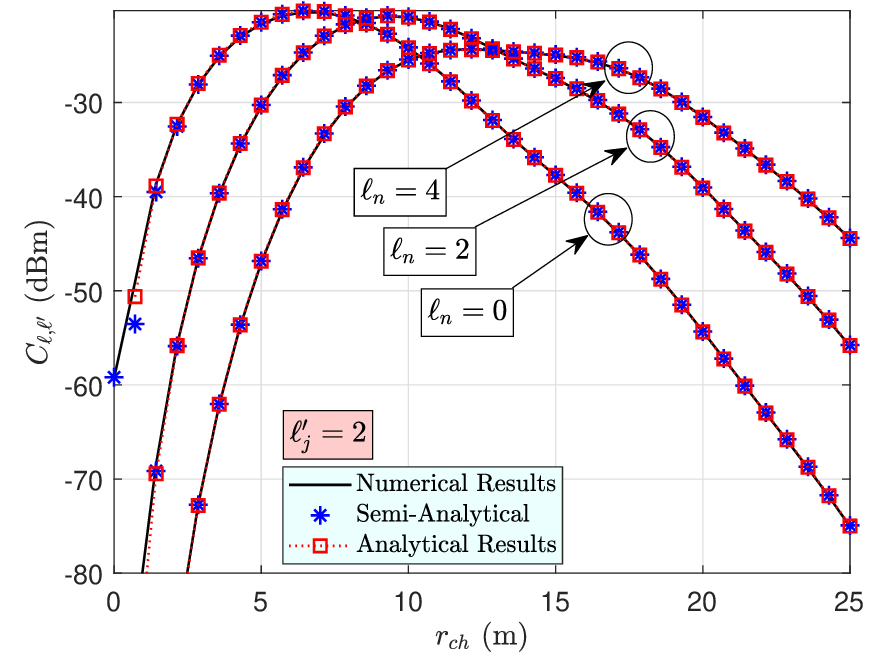}
		\label{cb2}
	}
	\hfill
	\subfloat[] {\includegraphics[width=2.3 in, height=1.4 in]{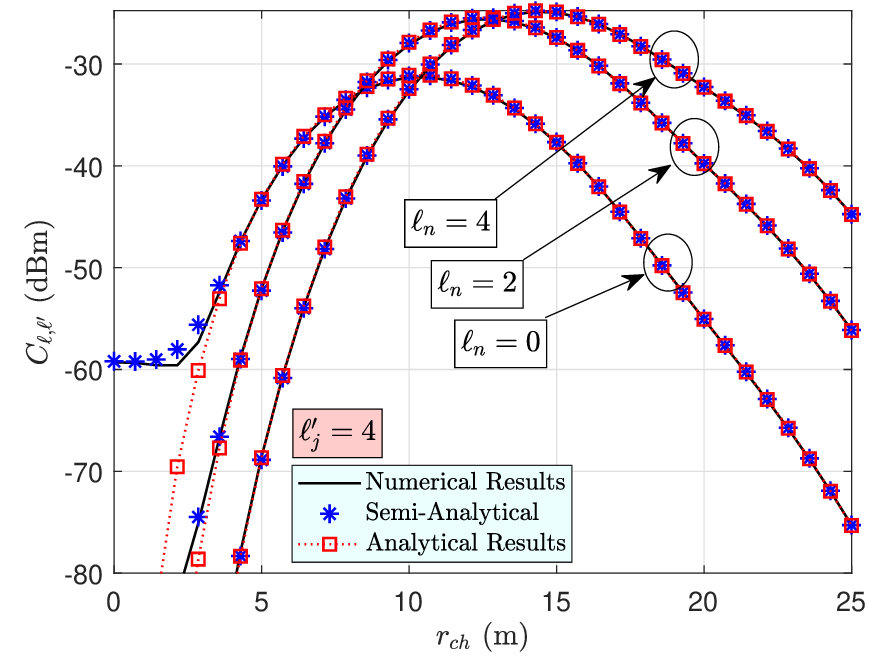}
		\label{cb3}
	}
	\caption{The crosstalk \( C_{\ell_n, \ell'_j} \) versus the pointing error \( r_{\text{ch}} \) for $p=0$, $w_0=2.5$ cm, and different values of \( (\ell_n, \ell'_j) \):  
		(a) \( \ell'_j = 0 \),  
		(b) \( \ell'_j = 2 \), and  
		(c) \( \ell'_j = 4 \).}
	\label{cb}
\end{figure*}

\section{Simulation Results}
In this paper, we provide a comprehensive analytical framework for the performance evaluation of OAM-based inter-satellite communication systems in the presence of pointing errors. Specifically, we first derive analytical expressions for the intermodal crosstalk coefficients \( C_{\ell_n, \ell'_j} \), which significantly accelerate the analysis compared to conventional Monte Carlo methods. Building on this foundation, we further develop a novel analytical expression for the average BER, enabling efficient and accurate assessment of system performance with minimal computational overhead.

To comprehensively evaluate the impact of pointing errors on OAM-based inter-satellite links, we conduct a two-stage simulation study. In the first stage, we analyze the behavior of intermodal crosstalk coefficients \( C_{\ell_n, \ell'_j} \) under various conditions to reveal key structural properties induced by pointing-induced misalignments. These insights play a central role in understanding the system's physical behavior and guiding the subsequent analysis. In the second stage, we leverage these results to evaluate the overall system performance using the proposed analytical BER expression, which is further validated through Monte Carlo simulations. The simulation setup and default parameters, representative of practical inter-satellite communication scenarios, are summarized in Table~\ref{Tab1}. These values are selected based on typical specifications reported in the optical wireless communication literature and the expected topology of LEO satellite constellations. To further examine the impact of various system parameters, any changes from these default values are clearly indicated in the corresponding figure captions.

\subsection{Crosstalk Analysis}
Initially, in Fig.~\ref{cf}, the crosstalk \( C_{\ell_n, \ell'_j} \) is plotted as a function of pointing error \( r_{\text{ch}} \) for for $p=0$, $w_0=5$ cm, and different mode combinations \( (\ell_n, \ell'_j) \in \{0, 2, 4\} \). Subfigures \ref{cf1}, \ref{cf2}, and \ref{cf3} correspond to the cases of \( \ell'_j = 0 \), \( \ell'_j = 2 \), and \( \ell'_j = 4 \), respectively.  
It should be noted that the legends in the figures refer to the following:
\begin{itemize} 
\item \textit{Numerical Results}: The exact two-dimensional solution for \( C_{\ell_n, \ell'_j} \) obtained using  \eqref{cros2},
\item \textit{Semi-Analytical}: The approximate one-dimensional integral solution derived in Proposition 1,  
\item \textit{Analytical}: The closed-form solution presented in Proposition 3.  
\end{itemize} 
As clearly shown in all three subfigures, for values of \( r_{\text{ch}} \) greater than 4 m, the results from Propositions 1 and 3 closely match the exact numerical results, accurately capturing the behavior of \( C_{\ell_n, \ell'_j} \). Only slight deviations can be observed for smaller \( r_{\text{ch}} \) values, particularly for the closed-form solution presented in Proposition 3.

One of the critical parameters influencing the performance of the OAM system under pointing errors is the beam width \( w(z) \), which is determined by the beam waist \( w_0 \) at the transmitter. To highlight the importance of \( w_0 \) and its impact on crosstalk, in Fig.~\ref{cb}, we repeat the results of Fig.~\ref{cf} for \( w_0 = 2.5 \, \text{cm} \).  
It should be noted that decreasing \( w_0 \) increases the beam width at the receiver, and therefore, we expect improved robustness to larger pointing errors. By comparing the results of Figs.~\ref{cf} and \ref{cb}, it can be observed that although the OAM system with \( w_0 = 5 \, \text{cm} \) performs better for smaller pointing errors, for larger pointing errors, the system with \( w_0 = 2.5 \, \text{cm} \) results in higher received power.  
For example, it can be seen that for \( w_0 = 2.5 \, \text{cm} \), the value \( C_{\ell_n, \ell'_j} \) for \( \ell_n = 4 \) and \( \ell'_j = 4 \) remains above \(-50 \, \text{dBm}\) even for \( r_{\text{ch}} = 25 \, \text{m} \), while for \( w_0 = 5 \, \text{cm} \), \( C_{\ell_n, \ell'_j} \) drops below \(-50 \, \text{dBm}\) at \( r_{\text{ch}} = 16 \, \text{m} \).

\begin{figure*}
	\centering
	\subfloat[] {\includegraphics[width=2.3 in, height=1.4 in]{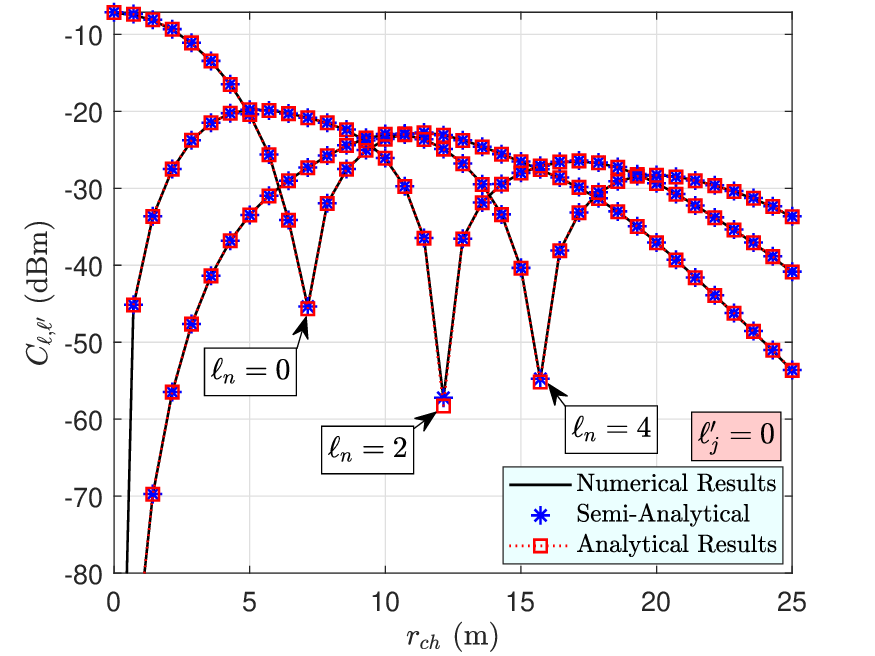}
		\label{cv1}
	}
	\hfill
	\subfloat[] {\includegraphics[width=2.3 in, height=1.4 in]{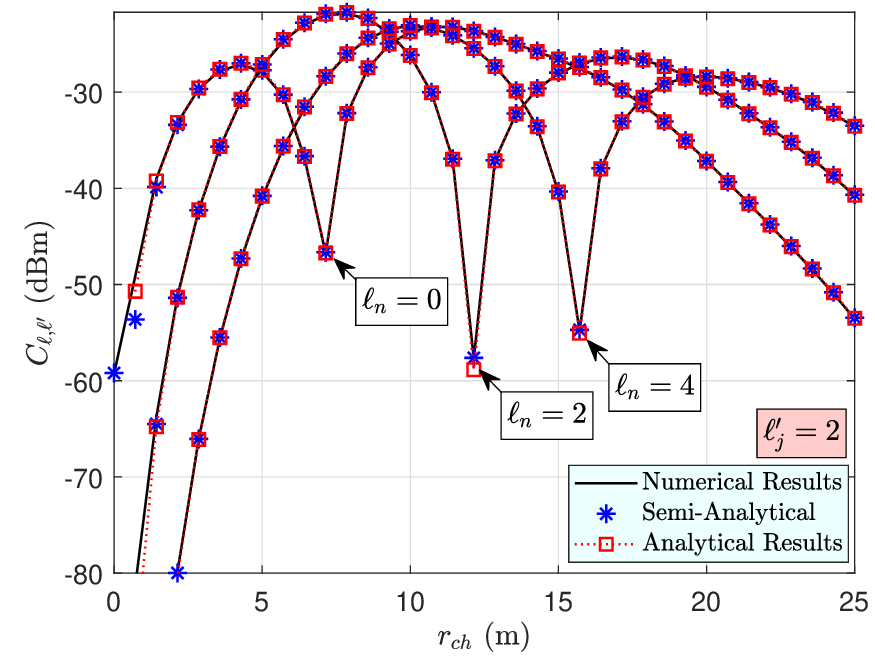}
		\label{cv2}
	}
	\hfill
	\subfloat[] {\includegraphics[width=2.3 in, height=1.4 in]{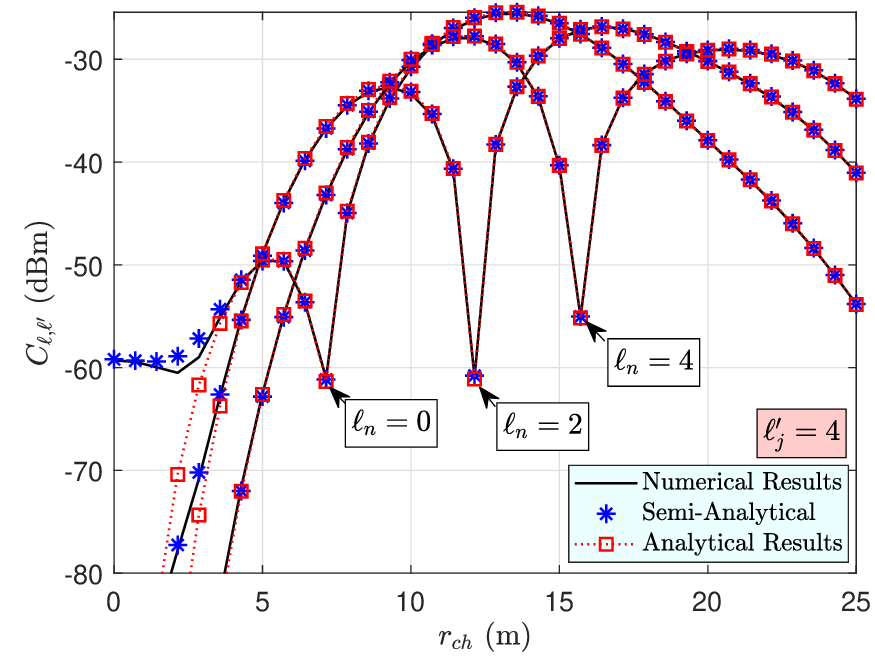}
		\label{cv3}
	}
	\caption{The crosstalk \( C_{\ell_n, \ell'_j} \) versus the pointing error \( r_{\text{ch}} \) for $p=1$, $w_0=2.5$ cm, and different values of \( (\ell_n, \ell'_j) \):  
		(a) \( \ell'_j = 0 \),  
		(b) \( \ell'_j = 2 \), and  
		(c) \( \ell'_j = 4 \).}
	\label{cv}
\end{figure*}

\begin{figure*}
	\centering
	\subfloat[] {\includegraphics[width=2.3 in, height=1.4 in]{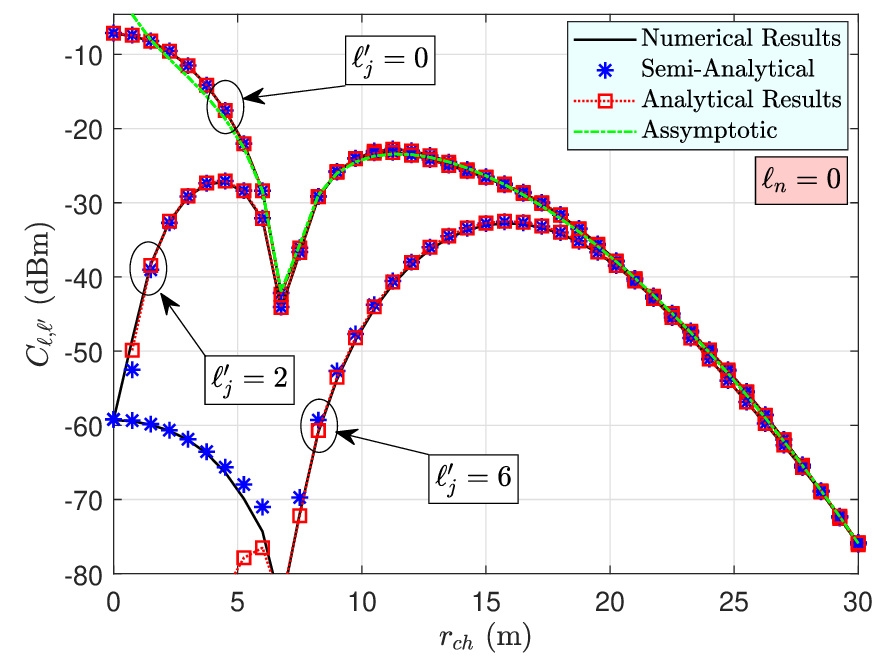}
		\label{cr1}
	}
	\hfill
	\subfloat[] {\includegraphics[width=2.3 in, height=1.4 in]{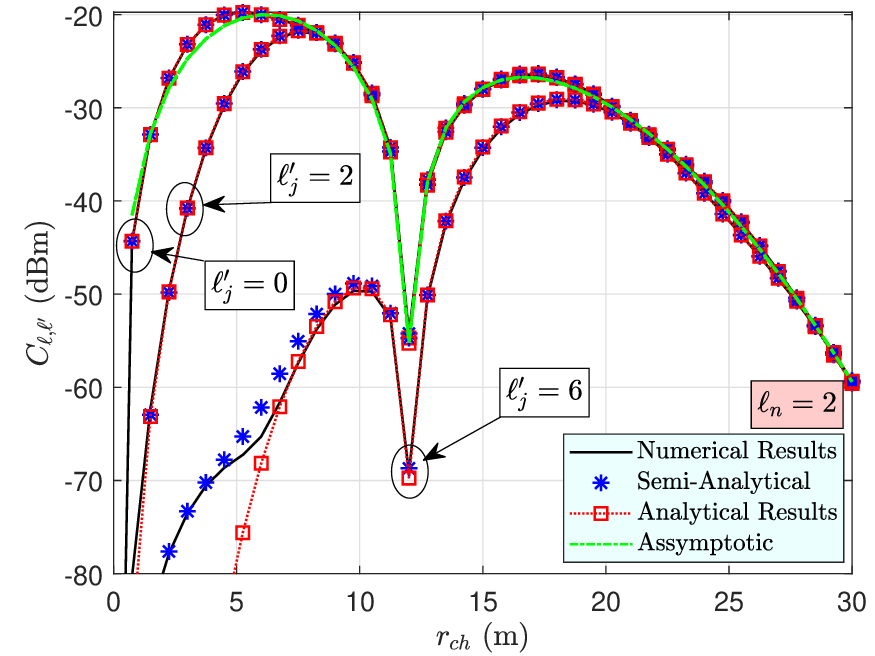}
		\label{cr2}
	}
	\hfill
	\subfloat[] {\includegraphics[width=2.3 in, height=1.4 in]{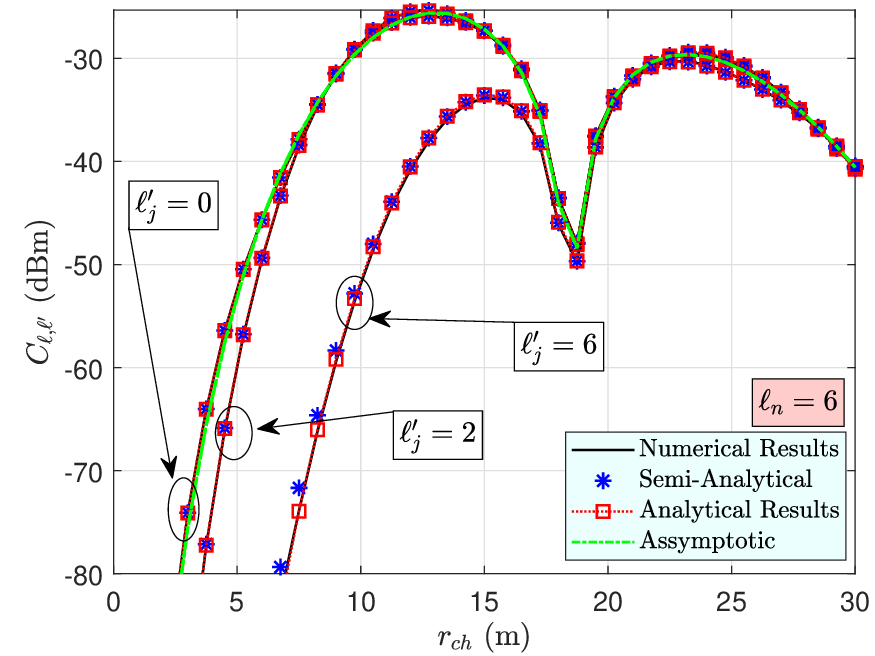}
		\label{cr3}
	}
	\caption{The crosstalk \( C_{\ell_n, \ell'_j} \) versus the pointing error \( r_{\text{ch}} \) for $p=1$, $w_0=2.5$ cm, and different values of \( (\ell_n, \ell'_j) \):  
		(a) \( \ell_n = 0 \),  
		(b) \( \ell_n = 2 \), and  
		(c) \( \ell_n = 6 \).}
	\label{cr}
\end{figure*}

Additionally, an important observation from both Figs.~\ref{cf} and \ref{cb} is that for small pointing errors, modes with smaller \( \ell_n \) exhibit better performance. In contrast, for severe pointing errors, modes with larger \( \ell_n \) perform better. This behavior is attributed to the radial characteristics of the LG beam pattern. As \( \ell_n \) increases, the doughnut-shaped intensity distribution of the LG mode shifts outward, resulting in a larger radial beam structure that becomes less sensitive to larger pointing errors.

Another parameter that influences the beam shape and its sensitivity to pointing errors is the radial index \( p \) of the LG beam pattern. To gain further insights, in Fig.~\ref{cv}, we repeat the results of Fig.~\ref{cb} for \( p = 1 \).  
As the results indicate, changing \( p \) significantly affects the values of \( C_{\ell_n, \ell'_j} \). It should be noted that by changing \( p \) from 0 to 1, the doughnut-shaped pattern of the LG beam undergoes a radial transformation. Instead of a single peak in the power distribution, the beam now exhibits two peaks along the radial direction. Consequently, the system behavior with respect to pointing errors also demonstrates two peaks in \( C_{\ell_n, \ell'_j} \).  
More importantly, despite the more complex behavior observed for \( p = 1 \), the analytical results presented in this paper closely follow the exact numerical results obtained from the two-dimensional integral. This demonstrates that the proposed analytical expressions can be effectively utilized for faster analysis and optimal design of the OAM system under consideration.

\begin{figure}
	\centering
	\subfloat[] {\includegraphics[width=1.65 in]{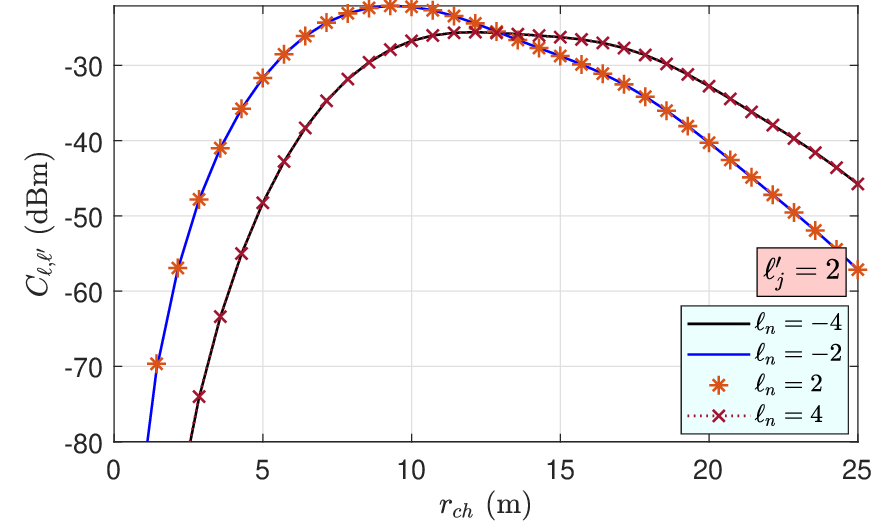}
		\label{cd1}	
		}
	\hfill
	\subfloat[] {\includegraphics[width=1.65 in]{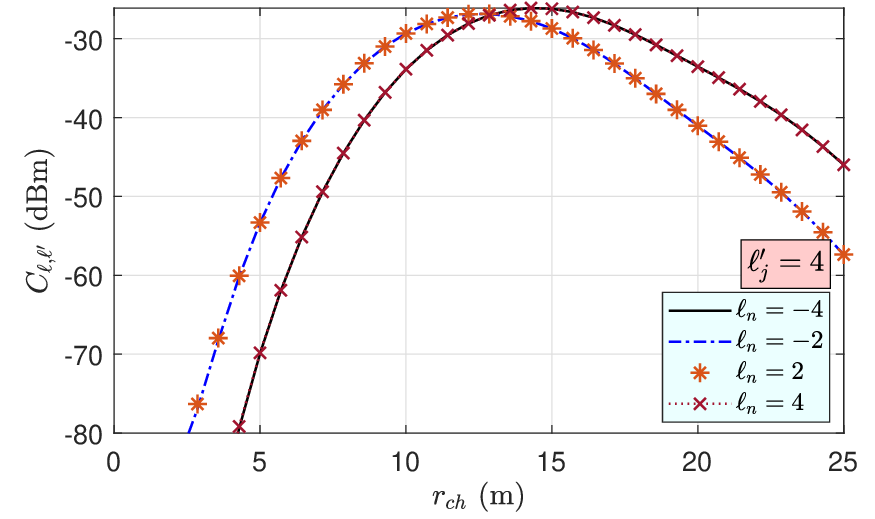}
		\label{cd2}
	}
	\caption{The crosstalk \( C_{\ell_n, \ell'_j} \) versus the pointing error \( r_{\text{ch}} \) for symmetric modes \( \ell_n = \{-4, -2, 2, 4\} \):  
			(a) \( \ell'_j = 2 \),  and
			(b) \( \ell'_j = 4 \). }
	\label{cd}
\end{figure}

In Fig.~\ref{cr}, the asymptotic behavior of crosstalk \( C_{\ell_n, \ell'_j} \) is investigated for large values of \( r_{\text{ch}} \). Unlike the previous figures, each subfigure in Fig.~\ref{cr} corresponds to a fixed value of \( \ell_n \) and varying values of \( \ell'_j \). It is important to note that \( \ell'_j \) represents the phase filter applied at the receiver to separate the \( \ell_n \) modes.  
The results clearly demonstrate that as \( r_{\text{ch}} \) increases, for any given \( \ell_n \), the output of all filters becomes similar. In other words, for large \( r_{\text{ch}} \), all phase information encoded in the OAM modes is lost, and only the radial mode information corresponding to \( \ell_n \) is preserved.  
Additionally, the results highlight that larger \( \ell'_j \) values exhibit better resilience in retaining phase information under pointing errors. This is an essential takeaway that should be considered when designing inter-satellite OAM links and determining achievable alignment accuracy. Moreover, this phenomenon can even be considered a security advantage for OAM-based links, as accessing phase information requires an eavesdropper to be precisely aligned in the propagation path, making eavesdropping significantly more challenging in practice.  
Another critical observation is that the asymptotic analytical behavior derived in Proposition~4 perfectly matches the results for \( \ell'_j = 0 \), with the behavior of other \( \ell'_j \) modes converging to it for larger \( r_{\text{ch}} \).

In Fig.~\ref{cd}, the effect of symmetric modes on crosstalk \( C_{\ell_n, \ell'_j} \) is analyzed. The results are obtained for \( \ell_n = \{-4, -2, 2, 4\} \), with subfigures \ref{cd1} and \ref{cd2} corresponding to \( \ell'_j = 2 \) and \( \ell'_j = 4 \), respectively.  
As demonstrated by the analytical results and clearly shown in Fig.~\ref{cd}, \( C_{\ell_n, \ell'_j} \simeq C_{-\ell_n, \ell'_j} \) holds true for long inter-satellite communication links. It should be noted that this result is valid only for long inter-satellite links, where the beam width at the receiver expands significantly and becomes much larger than the receiver aperture. This property does not hold for shorter terrestrial links.  
The results of this figure indicate that it is not possible to achieve orthogonal modes for \( \ell_n \) and \( -\ell_n \) in long inter-satellite OAM communication links.

\begin{figure}
	\begin{center}
		\includegraphics[width=2.3 in]{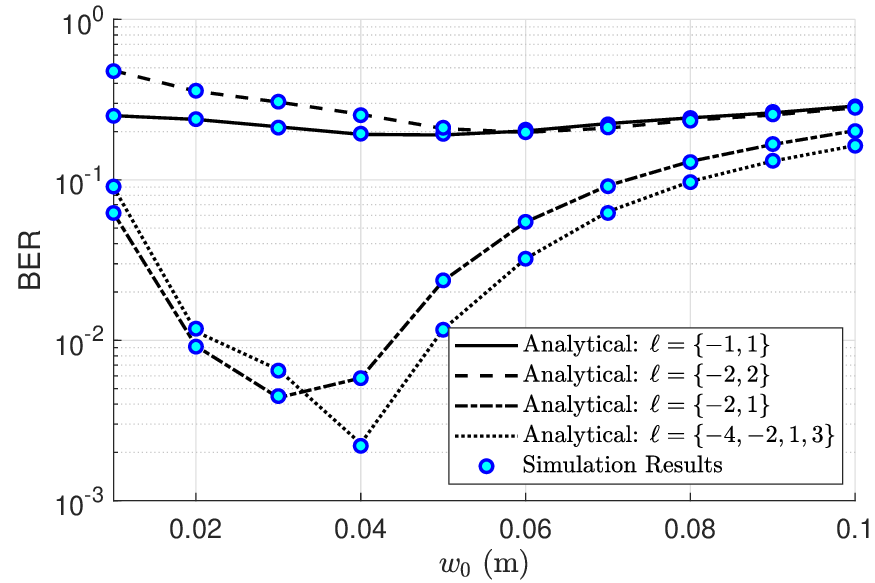}   
			\caption{Average BER versus transmitter beam waist $w_0$ for various OAM mode sets under pointing error. Asymmetric mode sets significantly outperform symmetric ones.}
		\label{fsh2}
	\end{center}
\end{figure}

\begin{figure}
	\centering
	\subfloat[] {\includegraphics[width=2.3 in]{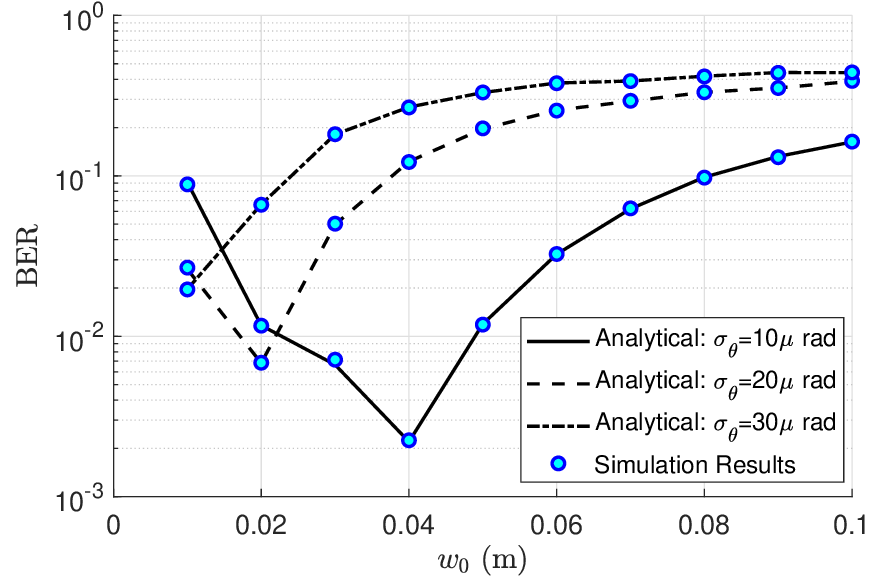}
		\label{cm1}	
	}
	\hfill
	\subfloat[] {\includegraphics[width=2.3 in]{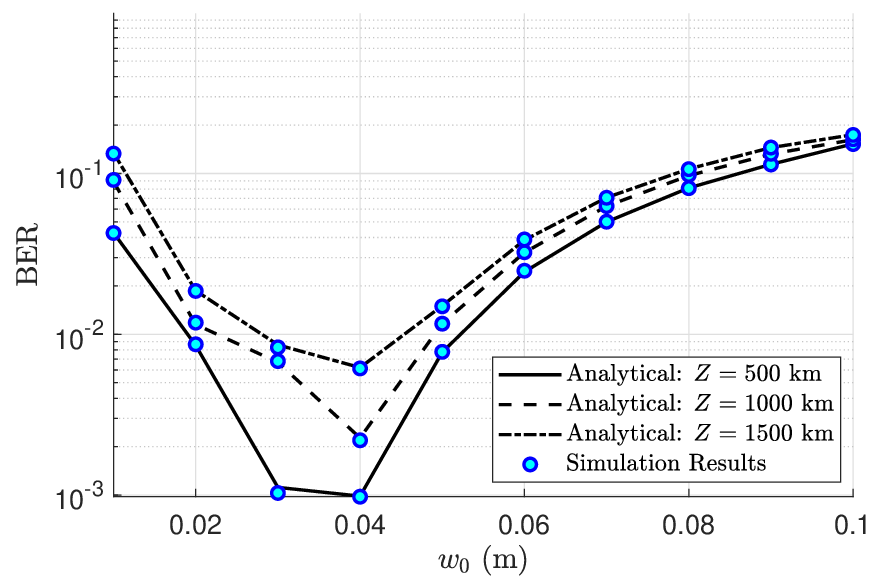}
		\label{cm2}
	}

		\caption{Average BER versus beam waist $w_0$ for (a)  different pointing error standard deviations $\sigma_\theta$, and (b) for various link distances $z$. Results correspond to the OAM mode set $\{-4, -2, 1, 3\}$.}

	\label{cm}
\end{figure}

\subsection{BER Analysis}
After conducting a comprehensive analysis and extracting detailed insights into the modeling of pointing errors and their impact on crosstalk in OAM-based inter-satellite links, we now investigate how pointing-induced crosstalk affects overall system performance in terms of BER.
It is important to highlight that this paper provides an analytical method for evaluating the BER, which relies solely on a one-dimensional integral to compute the average error probability. Using MATLAB on a standard desktop with Intel Core i7 CPU (2.8\,GHz, 16\,GB RAM), the proposed method delivers results in less than one second. 
To validate the accuracy of the derived analytical expressions, we compared the results with those obtained through a Monte Carlo simulation approach as presented in \cite{dabiri2024advancing}, where $10^7$ independent random data symbols are transmitted over $10^7$ randomly generated channel realizations. This simulation requires more than 20 minutes of runtime under the same computational setup.
As shown in Figs.~\ref{fsh2}, \ref{cm1}, and \ref{cm2}, the analytical results exhibit excellent agreement with the Monte Carlo simulation outcomes, while reducing the computational time by a factor exceeding 1200$\times$. This significant reduction is particularly important for inter-satellite communication, where satellites are in high-speed relative motion and links are continuously switching. As a result, the network topology dynamically changes, and link-dependent parameters such as divergence angle must be rapidly re-optimized and applied in near real-time.

Using the analytical BER expression developed in this paper, we first investigate the impact of mode selection on the performance of inter-satellite OAM communication systems under pointing errors, as depicted in Fig.~\ref{fsh2}. One of the most critical tunable parameters in satellite communication is the beam divergence angle, which is directly controlled by the transmitter beam waist $w_0$. Therefore, the BER performance is evaluated with respect to variations in $w_0$.
As shown in Fig.~\ref{fsh2}, symmetric mode pairs, such as $\{-2,2\}$ and $\{-1,1\}$, suffer from severe degradation in BER performance under pointing errors. This phenomenon can be explained based on a key observation from the first part of the simulations for inter-satellite communications, where the intermodal crosstalk coefficients approximately satisfy $C_{\ell_n, \ell'_j} \simeq C_{-\ell_n, \ell'_j}$. Consequently, the received signal vectors for symmetric modes become nearly indistinguishable, which increases the error probability. This is further confirmed by the analytical BER expression in \eqref{eq:ber_total_final_correct}, which simplifies as follows:
\begin{align} \label{eg1}
	&P_{\text{e,avg}|r_{\text{ch}}} \sim  \frac{1}{2} Q\left( \sqrt{ \frac{ \|\mathbf{h}_1-\mathbf{h}_2\|^2 }{4N_0} } \right),
\end{align}
where $\mathbf{h}_1 = \left[ \sqrt{C_{\ell_1,\ell'_1}}, \sqrt{C_{\ell_1,\ell'_2}} \right]^T$ and $\mathbf{h}_2 = \left[ \sqrt{C_{\ell_2,\ell'_1}}, \sqrt{C_{\ell_2,\ell'_2}} \right]^T$. Due to the approximate equality of $C_{\ell_n, \ell'_j}$ and $C_{-\ell_n, \ell'_j}$, the norm $\|\mathbf{h}_1-\mathbf{h}_2\|$ decreases, thus increasing $P_{\text{e,avg}|r_{\text{ch}}}$.
To mitigate this degradation, asymmetric mode pairs such as $\{-2,1\}$ can be employed, which (as illustrated in Fig.~\ref{fsh2}) significantly reduce the BER. This is a major insight of this study, which is supported by both analytical and simulation results. Furthermore, we also evaluated a 4-mode configuration using the set $\{-4, -2, 1, 3\}$, where two negative modes transmit one data stream and two positive modes transmit another. To ensure a fair comparison, the total transmit power in this 4-mode case was kept equal to that of the 2-mode scenario. Simulation results reveal that this configuration further improves BER performance. 
This improvement stems from the unique spatial distribution characteristics of OAM beams, which exhibit a doughnut-shaped intensity pattern. In such beams, lower-order modes tend to concentrate more energy near the beam center, while higher-order modes distribute their energy further from the center. This complementary distribution allows the modes to spatially diversify the signal energy, so that in the presence of pointing errors (which cause random beam misalignments), the combined use of inner (low-order) and outer (high-order) modes can collectively enhance the robustness of the system. As a result, the overall signal-to-noise ratio improves, leading to a lower bit error rate.

Next, Fig.~\ref{cm} examines the impact of two key physical parameters—tracking accuracy and link distance—on the BER performance of OAM-based inter-satellite systems. Motivated by the results in Fig.~\ref{fsh2}, these simulations are performed for the mode set $\{-4,-2,1,3\}$.
Specifically, Fig.~\ref{cm1} illustrates the BER versus $w_0$ for three different tracking accuracy values: $\sigma_\theta = 10$, $20$, and $30$~$\mu$rad. As expected, improving the tracking accuracy significantly reduces the BER. However, for each tracking condition, the system exhibits an optimal value of $w_0$ that minimizes the BER. Interestingly, low beam waist selection (e.g., $w_0<2$~cm) can negate the advantages of high tracking precision; for instance, a system with $\sigma_\theta = 10$~$\mu$rad and suboptimal $w_0$ may perform worse than one with $\sigma_\theta = 30$~$\mu$rad but properly tuned $w_0$. This behavior stems from the probabilistic misalignment of OAM's doughnut-shaped intensity nulls, emphasizing the necessity of adapting $w_0$ to both the tracking characteristics and the chosen mode set.

Fig.~\ref{cm2} evaluates the system BER for different link distances: $z = 500$, $1000$, and $1500$~km. Due to the high relative motion of satellites in a dynamic inter-satellite network, link distances vary rapidly, requiring frequent link switching and reconfiguration. As seen in the figure, BER performance degrades as the distance increases, which is primarily attributed to the amplification of pointing deviations over longer propagation paths. For instance, a 10~$\mu$rad tracking error results in a lateral beam offset of 5~m at 500~km, and 15~m at 1500~km. Nonetheless, the optimal $w_0$ values remain relatively close for all distances. This is because, while the misalignment increases with distance, the beam also naturally expands, partially mitigating the impact of pointing errors.

\section{Conclusion}
This study addresses the critical challenge of pointing errors in OAM-based inter-satellite communication systems by presenting a novel closed-form analytical framework for modeling intermodal crosstalk under realistic channel conditions, including beam width and receiver aperture size. Unlike traditional methods reliant on computationally intensive numerical simulations, the proposed approach offers a highly efficient and accurate means of evaluating system performance, enabling rapid design and optimization of OAM-based links. The analytical models are validated against numerical simulations, demonstrating their reliability and practical utility. Furthermore, the insights provided by this study highlight the intricate interplay between pointing errors, channel parameters, and system performance, offering valuable design guidelines for robust and scalable communication systems. By bridging theoretical modeling with practical application, this work contributes to the development of next-generation Low Earth Orbit (LEO) satellite constellations, facilitating ultra-high-capacity, low-latency, and globally connected networks.

\appendices
\section{} \label{AppA}
First, let us rewrite \eqref{cros2} as follows:
\begin{align} \label{cros3}
	&C_{\ell_n,\ell'_j}  =  A_1   \int_0^{r_a}  \Bigg| 
	\int_0^{2\pi}   
	\underbrace{\left(  \sqrt{ 2 A_2(r') } \right)^{|\ell_n|}    L_p^{|\ell_n|} \left( 2 A_2(r') \right) 
		e^{-A_2(r')} }_{T_1} \nonumber \\
	&\times \underbrace{e^{-i \ell_n \tan^{-1} \left( A_3(r') \right)
			+ i \ell_j' \phi'} }_{T_2} 
	\underbrace{e^{ -i k_\nu A_4(r') + i  \psi(z') } }_{T_3}
	\, d\phi'
	\Bigg|^2\, r'  \, dr'   
\end{align}
where
\begin{align}
	\begin{cases}
		&\!\!\! A_1 = \frac{  \eta G }{N_m^2 w^2(z')}  \frac{2 p!}{\pi (p + |\ell_n|)!},\\
		&\!\!\! A_2(r') = \frac{(r' \cos(\phi') + x_{\text{ch}})^2 + (r' \sin(\phi') + y_{\text{ch}})^2}{w^2(z')} \\
		&\!\!\! A_3(r') = \frac{r' \sin(\phi') + y_{\text{ch}}}{r' \cos(\phi') + x_{\text{ch}}} \\
		&\!\!\! A_4(r') = \frac{(r' \cos(\phi') + x_{\text{ch}})^2 + (r' \sin(\phi') + y_{\text{ch}})^2}{2 R(z')}.
	\end{cases}
\end{align}
The integrand in \eqref{cros3} consists of the product of three terms, \( T_1 \), \( T_2 \), and \( T_3 \), each of which is a function of \( r' \).

Note that for inter-satellite communications, the link length ranges from several hundred to several thousand kilometers. Therefore, the beam width is significantly larger compared to short-range links based on LG modes, which are typically on the order of a few centimeters or even less. The beam width is controlled by \( w_0 \), the beam waist at the transmitter. For a relatively large value of \( w_0 = 5 \, \text{cm} \) (resulting in a smaller \( w(z) \)), with a link length of \( 1000 \, \text{km} \), the beam width at the receiver exceeds \( 5 \, \text{m} \). This is still considered a very small value in practical inter-satellite communications. On the other hand, we have \( r' < r_a \), where \( r_a \) is on the order of a few centimeters. 

Therefore, it can be shown that all three terms \( A_1(r') \), \( A_2(r') \), and \( A_3(r') \) exhibit minimal variations with small changes in \( r' \) within the range \([0, r_a]\). Consequently, the terms \( T_1 \), \( T_2 \), and \( T_3 \) also undergo negligible changes. Thus, instead of integrating continuously, we can calculate \( C_{\ell_n,\ell'_j} \) over a limited number of discrete points \( \frac{k r_a}{K_r} \) for $k\in\{1,...,K_r\}$ and assume that the terms \( T_1 \), \( T_2 \), and \( T_3 \) remain approximately constant between \( \frac{k r_a}{K_r} \) and \( \frac{(k+1) r_a}{K_r} \). By adopting this approach, \( C_{\ell_n,\ell'_j} \) can be expressed as a one-dimensional integral as shown in \eqref{cros3}.

\begin{figure}
	\begin{center}
		\includegraphics[width=1.7 in]{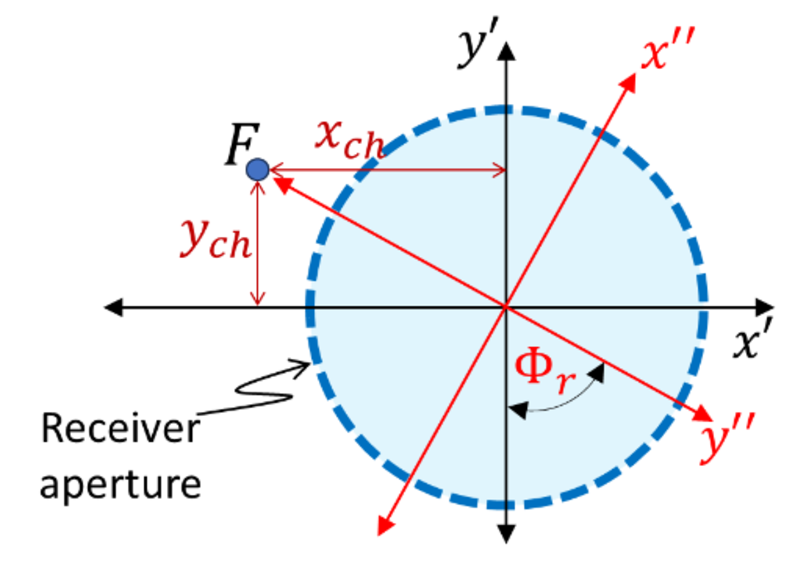}
		\caption{Definition of new coordinates \( (x'', y'') \), obtained by rotating the original coordinate system \( (x', y') \) by an angle \( \Phi_r \) such that the misalignment point $ r_{\text{ch}} = (x_{\text{ch}},y_{\text{ch}}) $ lies on the \( y'' \)-axis.}
		\label{nb2}
	\end{center}
\end{figure}
\section{} \label{AppB}
Without loss of generality, we define a new coordinate system in the receiver aperture, \( (x'', y'') \), as shown in Fig. \ref{nb2}. This new coordinate system is obtained by rotating the original coordinates \( (x', y') \) by an angle \( \Phi_r \), such that the beam misalignment center \( r_{\text{ch}} = (x_{\text{ch}}, y_{\text{ch}}) \) lies on the \( y'' \)-axis. Consequently, the new coordinate system in cylindrical coordinates \( (r'', \phi'') \) is defined such that \( r'' = r' \) and \( \phi'' = \phi' + \Phi_r \). By doing so, \eqref{cros2} can be rewritten as \eqref{crosp1}.

\begin{figure*}[!t]
	\normalsize
	\begin{align} \label{crosp1}
		&C_{\ell_n,\ell'_j}  =  \frac{  \eta G }{N_m^2 w^2(z')}  \int_0^{r_a}  \Bigg| 
		\int_0^{2\pi}   
		\sqrt{\frac{2 p!}{\pi (p + |\ell_n|)!}}  \left( \frac{\sqrt{2} \sqrt{  r'^2 + 2 r' r_{\text{ch}} \sin(\phi'') + r_{\text{ch}}^2    }}{w(z')} \right)^{|\ell_n|}  \\
		&\times L_p^{|\ell_n|} \left( \frac{2 \left( r'^2 + 2 r' r_{\text{ch}} \sin(\phi'') + r_{\text{ch}}^2 \right)}{w(z')^2} \right) 
		\exp\left(-\frac{r'^2 + 2 r' r_{\text{ch}} \sin(\phi'') + r_{\text{ch}}^2}{w(z')^2}\right) \nonumber \\
		&\times \exp\left({-i \ell_n \tan^{-1} \left( \frac{r' \sin(\phi'') + r_{\text{ch}}}{r' \cos(\phi'') } \right)
			+ i \ell_j' \phi''} \right) 
		\exp\left( -i k_\nu \frac{r'^2 + 2 r' r_{\text{ch}} \sin(\phi'') + r_{\text{ch}}^2}{2 R(z')} + i  \psi(z') \right)
		\, d\phi''
		\Bigg|^2\, r'  \, dr'  \nonumber 
	\end{align}
	\hrulefill
\end{figure*}	

Given the long link length and the tracking system error models (even for advanced tracking systems), the probability that \( r_{\text{ch}} = (x_{\text{ch}}, y_{\text{ch}}) \) is less than one meter is very low. On the other hand, since \( r' \leq r_a \), it is reasonable to assume with high accuracy that \( r_{\text{ch}} \gg r' \cos(\phi'') \) and \( r_{\text{ch}} \gg r' \sin(\phi'') \). 
Therefore, the term \( r'^2 + 2 r' r_{\text{ch}} \sin(\phi'') + r_{\text{ch}}^2 \) can be approximated as \( r_{\text{ch}}^2 \). Similarly, we have \( \left( \frac{r' \sin(\phi'') + r_{\text{ch}}}{r' \cos(\phi'')} \right) \gg 1 \), and thus \( \tan^{-1} \left( \frac{r' \sin(\phi'') + r_{\text{ch}}}{r' \cos(\phi'')} \right) \simeq \frac{\pi}{2} \). Based on this, \eqref{crosp1} can be approximated as shown in \eqref{cros5}.

\begin{figure*}[!t]
	\normalsize
	\begin{align} \label{cros5}
		&C_{\ell_n,\ell'_j}  =  A_1  \int_0^{r_a}  \Bigg| 
		\int_0^{2\pi}   
		\underbrace{ \left( \frac{\sqrt{2} r_{\text{ch}} }{w(z')} \right)^{|\ell_n|}  
			L_p^{|\ell_n|} \left( \frac{2 r^2_{\text{ch}} }{w(z')^2} \right) 
			\exp\left(-\frac{r_{\text{ch}}^2}{w(z')^2}\right) 
			\exp\left({-i \frac{\ell_n \pi}{2} + i  \psi(z')} \right) }_{T_4} \nonumber \\
		&\times  \exp\left( -i k_\nu \frac{r'^2 + 2 r' r_{\text{ch}} \sin(\phi'') + r_{\text{ch}}^2}{2 R(z')}  \right)
		\exp\left({ i \ell_j' \phi'' } \right) 
		\, d\phi''
		\Bigg|^2\, r'  \, dr'  
	\end{align}
	\hrulefill
\end{figure*}

Note that in \eqref{cros5}, for the term \( \exp\left( -i k_\nu \frac{r'^2 + 2 r' r_{\text{ch}} \sin(\phi'') + r_{\text{ch}}^2}{2 R(z')} \right) \), the mentioned approximations are not applicable because \( k_\nu \), the wave number, is a large value and is highly sensitive to such approximations. As can be observed, the term \( T_4 \) in \eqref{cros5} is independent of \( \phi'' \), and therefore \eqref{cros5} can be simplified to \eqref{cros6}.
\begin{figure*}[!t]
	\normalsize
	\begin{align} \label{cros6}   \resizebox{0.93\hsize}{!}{$
		C_{\ell_n,\ell'_j}  =  A_1       
		\left[ \left( \frac{\sqrt{2} r_{\text{ch}} }{w(z')} \right)^{|\ell_n|}  
		L_p^{|\ell_n|} \left( \frac{2 r^2_{\text{ch}} }{w(z')^2} \right) 
		\exp\left(-\frac{r_{\text{ch}}^2}{w(z')^2}\right)  \right]^2  
		  \int_0^{r_a} \Bigg| \underbrace{
			\int_0^{2\pi}  \exp\left( -i k_\nu \frac{r'^2 + 2 r' r_{\text{ch}} \sin(\phi'') + r_{\text{ch}}^2}{2 R(z')}  \right)  
			\exp\left({ i \ell_j' \phi'' } \right) 
			\, d\phi''  }_{T_5}
		\Bigg|^2\, r'  \, dr'  $}
	\end{align}
	\hrulefill
\end{figure*}
The integral \( T_5 \) in \eqref{cros6} can be expressed as follows:
\begin{align}
	\label{po1}
	&T_5 = \exp\left( -i k_\nu \frac{r'^2  + r_{\text{ch}}^2}{2 R(z')}  \right)  
	\exp\left({ i \ell_j' \phi'' } \right) \nonumber \\
	& \times \underbrace{\int_0^{2\pi}  \exp\left( -i  \frac{2 k_\nu r' r_{\text{ch}}  }{2 R(z')} \sin(\phi'') \right)  
		\exp\left({ i \ell_j' \phi'' } \right) 
		\, d\phi''}_{T_6}
\end{align}
Using \cite[Eq. (8.511.4)]{gradshteyn2007table}, \( \exp\left( -i \frac{2 k_\nu r' r_{\text{ch}}}{2 R(z')} \sin(\phi'') \right) \) can be expanded as follows:
\begin{align} \label{po2}
	&\exp\left( -i \frac{2 k_\nu r' r_{\text{ch}}}{2 R(z')} \sin(\phi'') \right) 
	=   \nonumber \\
	&~~~~~~~~~~ \sum_{n=-\infty}^\infty J_n\left( \frac{2 k_\nu r' r_{\text{ch}}}{2 R(z')} \right) \exp\left( i n \phi'' \right),
\end{align}
where \( J_n(\cdot) \) represents the Bessel function of the first kind of order \( n \).
By substituting \eqref{po2} into the integral, \( T_6 \) is obtained as:
\begin{align} \label{po3}
	T_6 = \sum_{n=-\infty}^\infty J_n\left( \frac{2 k_\nu r' r_{\text{ch}}}{2 R(z')} \right)
	\underbrace{ \int_0^{2\pi}  \exp\left( i (n+\ell_j') \phi'' \right) \, d\phi''}_{T_7}.
\end{align}
The integral $T_7$ is obtained as:
\begin{align} \label{po4}
	\int_0^{2\pi} \exp\left( i (n + \ell_j') \phi'' \right) \, d\phi'' = 
	\begin{cases} 
		2\pi, & \text{if } n + \ell_j' = 0, \\
		0, & \text{otherwise}.
	\end{cases}
\end{align}
Finally, by substituting \eqref{po1}, \eqref{po3}, and \eqref{po4} into \eqref{cros6}, and after a series of simplifications, we arrive at \eqref{cros8}.

\section{} \label{AppC}
To derive the closed form, we utilize \eqref{cros8} from Proposition 2. Following the reasoning provided in Appendix A for Proposition 1, it can be shown that the term \( r' \left( J_{\ell_j'}\left( \frac{k r_{\text{ch}}}{R(z')} r' \right) \right)^2 \) exhibits minimal variation for \( r' < r_a \). Therefore, the range \( [0,r_a] \) can be divided into \( K_r \) subintervals, and it can be demonstrated that the term \( r' \left( J_{\ell_j'}\left( \frac{k r_{\text{ch}}}{R(z')} r' \right) \right)^2 \) remains approximately constant between \( \frac{r_a k}{K_r} \) and \( \frac{r_a (k+1)}{K_r} \). Consequently, \eqref{cros8} can be expressed in the closed form \eqref{cros9}, resulting from the summation of \( K_r \) distinct terms.

\section{} \label{AppD}
Let us define $x_k = \frac{k_\nu  r_{\text{ch}} r_a}{R(z') K_r} k$.
The Bessel function can be approximated as \cite{abramowitz1948handbook}:
\begin{align}
	J_{\ell_j'}(x_k) \approx \sqrt{\frac{2}{\pi x_k}} \cos\left(x_k - \frac{\ell_j' \pi}{2} - \frac{\pi}{4}\right).
\end{align}
Squaring this approximation:
\begin{align} \label{s1}
	(J_{\ell_j'}(x_k))^2 &\simeq \frac{2}{\pi x_k} \cos^2\left(x_k - \frac{\ell_j' \pi}{2} - \frac{\pi}{4}\right) \nonumber \\
	& = \frac{2}{\pi x_k} \cdot \frac{1}{2}\bigl[1 + \cos(2x_k - \ell_j' \pi - \tfrac{\pi}{2})\bigr] \nonumber \\
	&= \frac{1}{\pi x_k}\bigl[1 + (-1)^{\ell'_j}\sin(2x_k )\bigr]  \simeq \frac{1}{\pi x_k}.
\end{align}
Using \eqref{s1}, the asymptotic expression for \( C_{\ell_n,\ell'_j} \) is obtained in \eqref{cros10}.

\balance

\end{document}